\documentclass[12pt]{amsart}
\usepackage{}
\usepackage{amsmath}
\usepackage{amsfonts}
\usepackage{amssymb}
\usepackage[all,cmtip]{xy}           

\usepackage{bbding}
\usepackage{txfonts}
\usepackage[shortlabels]{enumitem}
\usepackage{ifpdf}
\ifpdf
\usepackage[colorlinks,final,backref=page,hyperindex]{hyperref}
\else
\usepackage[colorlinks,final,backref=page,hyperindex,hypertex]{hyperref}
\fi
\usepackage{tikz}
\usepackage[active]{srcltx}

\usepackage{difftrees}
\makeatletter

\newtheorem{theorem}{Theorem}[section]
\newtheorem{prop}[theorem]{Proposition}
\newtheorem{lemma}[theorem]{Lemma}
\newtheorem{coro}[theorem]{Corollary}
\newtheorem{prop-def}{Proposition-Definition}[section]

\theoremstyle{definition}
\newtheorem{defn}[theorem]{Definition}

\newtheorem{remark}[theorem]{Remark}
\newtheorem{exam}[theorem]{Example}

\newcommand{\nc}{\newcommand}

\newcommand {\emptycomment}[1]{}

\nc{\delete}[1]{{}}
\nc{\mmargin}[1]{}

\nc{\mlabel}[1]{\label{#1}}  
\nc{\mcite}[1]{\cite{#1}}  
\nc{\mref}[1]{\ref{#1}}  
\nc{\meqref}[1]{\eqref{#1}}  
\nc{\mbibitem}[1]{\bibitem{#1}} 

\delete{
	\nc{\mlabel}[1]{\label{#1}  
		{\hfill \hspace{1cm}{\bf{{\ }\hfill(#1)}}}}
	\nc{\mcite}[1]{\cite{#1}{{\bf{{\ }(#1)}}}}  
	\nc{\mref}[1]{\ref{#1}{{\bf{{\ }(#1)}}}}  
	\nc{\meqref}[1]{\eqref{#1}{{\bf{{\ }(#1)}}}}  
	\nc{\mbibitem}[1]{\bibitem[\bf #1]{#1}} 
}

\newcommand{\g}{\mathfrak h}

\newcommand{\h}{\mathfrak k}
\newcommand{\huaS}{\mathcal{S}}
\newcommand{\huaH}{\mathcal{H}}
\newcommand{\huaK}{\mathcal{K}}

 \font\cyrs=wncyr7

\newcommand{\bk}{{\mathbf{k}}}

\nc{\vep}{\varepsilon}
\nc{\bin}[2]{ (_{\stackrel{\scs{#1}}{\scs{#2}}})}  
\nc{\binc}[2]{(\!\! \begin{array}{c} \scs{#1}\\
		\scs{#2} \end{array}\!\!)}  
\nc{\bincc}[2]{  ( {\scs{#1} \atop
		\vspace{-1cm}\scs{#2}} )}  
\nc{\oline}[1]{\overline{#1}}
\nc{\mapm}[1]{\lfloor\!|{#1}|\!\rfloor}
\nc{\bs}{\bar{S}}
\nc{\cast}{{\,\mbox{\raisebox{.8pt}{$\scriptstyle \circledast$}}\,}}
\nc{\la}{\longrightarrow}
\nc{\ot}{\otimes}
\nc{\rar}{\rightarrow}
\nc{\dar}{\downarrow}
\nc{\dap}[1]{\downarrow \rlap{$\scriptstyle{#1}$}}
\nc{\defeq}{\stackrel{\rm def}{=}}
\nc{\dis}[1]{\displaystyle{#1}}
\nc{\dotcup}{\ \displaystyle{\bigcup^\bullet}\ }
\nc{\hcm}{\ \hat{,}\ }
\nc{\hts}{\hat{\otimes}}
\nc{\hcirc}{\hat{\circ}}
\nc{\lleft}{[}
\nc{\lright}{]}
\nc{\curlyl}{\left \{ \begin{array}{c} {} \\ {} \end{array}
	\right .  \!\!\!\!\!\!\!}
\nc{\curlyr}{ \!\!\!\!\!\!\!
	\left . \begin{array}{c} {} \\ {} \end{array}
	\right \} }
\nc{\longmid}{\left | \begin{array}{c} {} \\ {} \end{array}
	\right . \!\!\!\!\!\!\!}
\nc{\ora}[1]{\stackrel{#1}{\rar}}
\nc{\ola}[1]{\stackrel{#1}{\la}}
\nc{\scs}[1]{\scriptstyle{#1}} \nc{\mrm}[1]{{\rm #1}}
\nc{\dirlim}{\displaystyle{\lim_{\longrightarrow}}\,}
\nc{\invlim}{\displaystyle{\lim_{\longleftarrow}}\,}
\nc{\dislim}[1]{\displaystyle{\lim_{#1}}} \nc{\colim}{\mrm{colim}}
\nc{\mvp}{\vspace{0.3cm}} \nc{\tk}{^{(k)}} \nc{\tp}{^\prime}
\nc{\ttp}{^{\prime\prime}} \nc{\svp}{\vspace{2cm}}
\nc{\vp}{\vspace{8cm}}
\nc{\modg}[1]{\!<\!\!{#1}\!\!>}
\nc{\intg}[1]{F_C(#1)}
\nc{\lmodg}{\!<\!\!}
\nc{\rmodg}{\!\!>\!}
\nc{\cpi}{\widehat{\Pi}}
\nc{\ssha}{{\mbox{\cyrs X}}} 
\nc{\tsha}{{\mbox{\cyrt X}}}
\nc{\shpr}{\diamond}    
\nc{\labs}{\mid\!}
\nc{\rabs}{\!\mid}

\nc{\ad}{\mrm{ad}}
\nc{\rRB}{\mathsf{rRB}}
\nc{\cocrRB}{\mathsf{cocrRB}}
\nc{\PH}{\mathsf{PH}}
\nc{\cocPH}{\mathsf{cocPH}}
\nc{\ann}{\mrm{ann}}
\nc{\Aut}{\mrm{Aut}}
\nc{\Der}{\mrm{Der}}
\nc{\Sym}{\mrm{Sym}}
\nc{\br}{\mrm{bre}}
\nc{\can}{\mrm{can}}
\nc{\Cont}{\mrm{Cont}}
\nc{\rchar}{\mrm{char}}
\nc{\cok}{\mrm{coker}}
\nc{\de}{\mrm{dep}}
\nc{\dtf}{{R-{\rm tf}}}
\nc{\dtor}{{R-{\rm tor}}}

\nc{\Dif}{\mrm{Diff}}
\nc{\Div}{\mrm{Div}}
\nc{\End}{\mrm{End}}
\nc{\Ext}{\mrm{Ext}}
\nc{\Fil}{\mrm{Fil}}
\nc{\Fr}{\mrm{Fr}}
\nc{\Frob}{\mrm{Frob}}
\nc{\Gal}{\mrm{Gal}}
\nc{\GL}{\mrm{GL}}
\nc{\Gr}{\mrm{Gr}}
\nc{\Hom}{\mrm{Hom}}
\nc{\Hoch}{\mrm{Hoch}}
\nc{\hsr}{\mrm{H}}
\nc{\hpol}{\mrm{HP}}
\nc{\id}{\mrm{id}}
\nc{\im}{\mrm{im}}
\nc{\inv}{\mrm{inv}}
\nc{\Id}{\mrm{Id}}
\nc{\ID}{\mrm{ID}}
\nc{\Irr}{\mrm{Irr}}
\nc{\incl}{\mrm{incl}}
\nc{\length}{\mrm{length}}
\nc{\NLSW}{\mrm{NLSW}}
\nc{\Lie}{\mrm{Lie}}
\nc{\mchar}{\rm char}
\nc{\mpart}{\mrm{part}}
\nc{\ql}{{\QQ_\ell}}
\nc{\qp}{{\QQ_p}}
\nc{\rank}{\mrm{rank}}
\nc{\rcot}{\mrm{cot}}
\nc{\rdef}{\mrm{def}}
\nc{\rdiv}{{\rm div}}
\nc{\rtf}{{\rm tf}}
\nc{\rtor}{{\rm tor}}
\nc{\res}{\mrm{res}}
\nc{\SL}{\mrm{SL}}
\nc{\Spec}{\mrm{Spec}}
\nc{\tor}{\mrm{tor}}
\nc{\Tr}{\mrm{Tr}}
\nc{\tr}{\mrm{tr}}
\nc{\wt}{\mrm{wt}}

\newcommand{\frki}{\mathfrak i}

\newcommand{\co}{\mathsf{cosh}}

\newcommand{\huaX}{\mathcal{X}}
\newcommand{\huaY}{\mathcal{Y}}

\nc{\bfk}{{\bf k}}
\nc{\bfone}{{\bf 1}}
\nc{\bfzero}{{\bf 0}}
\nc{\detail}{\marginpar{\bf More detail}
	\noindent{\bf Need more detail!}
	\svp}
\nc{\gap}{\marginpar{\bf Incomplete}\noindent{\bf Incomplete!!}
	\svp}
\nc{\FMod}{\mathbf{FMod}}
\nc{\Int}{\mathbf{Int}}
\nc{\Mon}{\mathbf{Mon}}
\nc{\remarks}{\noindent{\bf Remarks: }}
\nc{\Rep}{\mathbf{Rep}}
\nc{\Rings}{\mathbf{Rings}}
\nc{\Sets}{\mathbf{Sets}}
\nc{\Diff}{\mathbf{Diff}}
\nc{\Inte}{\mathbf{Inte}}
\nc{\U}{\mathbf{U}}

\nc{\BA}{{\mathbb A}}   \nc{\CC}{{\mathbb C}}
\nc{\DD}{{\mathbb D}}   \nc{\EE}{{\mathbb E}}
\nc{\FF}{{\mathbb F}}   \nc{\GG}{{\mathbb G}}
\nc{\HH}{{\mathbb H}}   \nc{\LL}{{\mathbb L}}
\nc{\NN}{{\mathbb N}}   \nc{\PP}{{\mathbb P}}
\nc{\QQ}{{\mathbb Q}}   \nc{\RR}{{\mathbb R}}
\nc{\TT}{{\mathbb T}}   \nc{\VV}{{\mathbb V}}
\nc{\ZZ}{{\mathbb Z}}   \nc{\TP}{\widetilde{P}}
\newcommand{\huaO}{{\mathcal{OT}}}

\newcommand{\huaT}{{\mathcal{T}}}

\nc{\cala}{{\mathcal A}}    \nc{\calc}{{\mathcal C}}
\nc{\cald}{\mathcal{D}}     \nc{\cale}{{\mathcal E}}
\nc{\calf}{{\mathcal F}}    \nc{\calg}{{\mathcal G}}
\nc{\calh}{{\mathcal H}}    \nc{\cali}{{\mathcal I}}
\nc{\call}{{\mathcal L}}    \nc{\calm}{{\mathcal M}}
\nc{\caln}{{\mathcal N}}    \nc{\calo}{{\mathcal O}}
\nc{\calp}{{\mathcal P}}    \nc{\calr}{{\mathcal R}}

\nc{\cals}{{\mathcal S}}    \nc{\calt}{{\Omega}}
\nc{\calv}{{\mathcal V}}    \nc{\calw}{{\mathcal W}}
\nc{\calx}{{\mathcal X}}

\nc{\fraka}{{\mathfrak a}}
\nc{\frakb}{\mathfrak{b}}
\nc{\frakg}{{\frak g}}
\nc{\frakl}{{\frak l}}
\nc{\fraks}{{\frak s}}
\nc{\frakB}{{\frak B}}
\nc{\frakm}{{\frak m}}
\nc{\frakM}{{\frak M}}
\nc{\frakp}{{\frak p}}
\nc{\frakW}{{\frak W}}
\nc{\frakX}{{\frak X}}
\nc{\frakS}{{\frak S}}
\nc{\frakA}{{\frak A}}
\nc{\frakx}{{\frakx}}

\nc{\ynr}[1]{\textcolor{orange}{\underline{Yunnan:}#1 }}

\nc{\lir}[1]{\textcolor{red}{\underline{Li:}#1 }}

\delete{
	\input cyracc.def

	\newtheorem{theorem}{Theorem}[section]
	\newtheorem{lemma}[theorem]{Lemma}

	\theoremstyle{definition}

	\theoremstyle{remark}
	\newtheorem{remark}[theorem]{Remark}

	\allowdisplaybreaks
	
	\numberwithin{equation}{section}
	
}

\begin{document}

\title[Post-Hopf algebras, relative Rota-Baxter operators and solutions of YBE]{Post-Hopf algebras, relative Rota-Baxter operators and solutions of the Yang-Baxter equation}

\author{Yunnan Li}
\address{School of Mathematics and Information Science, Guangzhou University,
Guangzhou 510006, China}
\email{ynli@gzhu.edu.cn}

\author{Yunhe Sheng}
\address{Department of Mathematics, Jilin University, Changchun 130012, Jilin, China}
\email{shengyh@jlu.edu.cn}

\author{Rong Tang}
\address{Department of Mathematics, Jilin University, Changchun 130012, Jilin, China}
\email{tangrong@jlu.edu.cn}

\begin{abstract}
In this paper, first we introduce the notion of a post-Hopf algebra, which gives rise to a post-Lie algebra on the space of primitive elements and there is naturally a post-Hopf algebra structure on the universal enveloping algebra of a post-Lie algebra. A novel property is that a cocommutative post-Hopf algebra gives rise to a generalized Grossman-Larsson product, which leads to a subadjacent Hopf algebra and can be used to construct solutions of the  Yang-Baxter equation. Then we introduce the notion of relative Rota-Baxter operators on   Hopf algebras.
A cocommutative post-Hopf algebra gives rise to a relative Rota-Baxter operator on its subadjacent Hopf algebra. Conversely,  a relative Rota-Baxter operator also induces a post-Hopf algebra. Then we show that relative Rota-Baxter operators give rise to matched pairs of Hopf algebras. Consequently, post-Hopf algebras and relative Rota-Baxter operators give solutions of the Yang-Baxter equation in certain cocommutative Hopf algebras. Finally we characterize   relative Rota-Baxter operators on Hopf algebras using relative Rota-Baxter operators on the Lie algebra of primitive elements, graphs and  module bialgebra structures.
\end{abstract}

\keywords{post-Hopf algebra, Hopf algebra, relative Rota-Baxter operator,  Yang-Baxter equation, matched pair\\
\qquad 2020 Mathematics Subject Classification. 16T05, 16T25,   17B38}

\maketitle

\tableofcontents

\allowdisplaybreaks

\section{Introduction}

The Yang-Baxter equation is an important subject in mathematical physics \cite{Ya}. Drinfeld
highlighted  the importance of the study of set-theoretical solutions of the
Yang-Baxter equation in \cite{Drinfeld}. The pioneer works on set-theoretical solutions are those of Etingof-Schedler-Soloviev \cite{Etingof},   Lu-Yan-Zhu \cite{Lu} and Gateva-Ivanova-Van den Bergh \cite{GVan}.  To understand  the structure of set-theoretical solutions, Rump introduced braces in \cite{Rump1} for abelian groups, which provide involutive nondegenerate solutions. See also \cite{CJO,CJO-1} for more details about the finite simple solutions of the
Yang-Baxter equation. Later Guarnieri and Vendramin generalized braces to the nonabelian case and introduced  skew braces in \cite{GV}, which provide nondegenerate set-theoretical solutions of the Yang-Baxter equation. Recently, Gateva-Ivanova  \cite{G}   used braided groups and
braces to study set-theoretical solutions of the
Yang-Baxter equation. In \cite{AGV}, Angiono, Galindo and Vendramin introduced the notion of Hopf braces, generalizing Rump's braces and Guarnieri-Vendramin's skew-braces. Any Hopf brace produces a solution of the  Yang-Baxter equation.

In this paper, we provide another approach to understand the structure of set-theoretical solutions of the   Yang-Baxter equation in certain Hopf algebras. In particular, we introduce the notion of   post-Hopf algebras, which naturally provide solutions of the   Yang-Baxter equation in the underlying vector spaces. We also introduce the notion of relative Rota-Baxter operators on  Hopf algebras, which naturally give rise to post-Hopf algebras, and thus to solutions of the  Yang-Baxter equation. The whole theory is based on the fact that a   cocommutative  post-Hopf algebra gives rise to a  generalized Grossman-Larsson product, which leads to a subadjacent Hopf algebra. Note that the classical  Grossman-Larsson product was defined in the context of polynomials of ordered rooted trees \cite{MW}, and have important applications in the studies of Magnus expansions \cite{CP,EM} and Lie-Butcher series  \cite{ML,MW}.

A post-Hopf algebra is a Hopf algebra $H$ equipped with a coalgebra homomorphism from $H\otimes H$ to $H$ satisfying some compatibility conditions (see Definition \ref{defi:pH}).  Magma algebras, in particular ordered rooted trees, provide a class of examples of post-Hopf algebras.  A  cocommutative  post-Hopf algebra gives rise to a new subadjacent Hopf algebra and a module bialgebra structure on itself. The terminology of post-Hopf algebras is justified by the fact that a post-Hopf algebra gives rise to a post-Lie algebra on the space of primitive elements. The notion of post-Lie algebras was introduced in \cite{Val}, and have important  applications in geometric numerical integration \cite{Fard1,Fard2}. In \cite{ELM}, Ebrahimi-Fard,  Lundervold and   Munthe-Kaas studied the Lie enveloping algebra of a post-Lie algebra, which turns out to be a post-Hopf algebra. They also find that there is a new Hopf algebra structure (the subadjacent Hopf algebra) on the Lie enveloping algebra of a post-Lie algebra, by which the Magnus expansions and Lie-Butcher series can be constructed. The subadjacent Hopf algebra  is also the main ingredient in our construction of solutions of the   Yang-Baxter equation. Moreover, we show that cocommutative post-Hopf algebras and cocommutative Hopf braces are equivalent. As a byproduct, we obtain the notion of pre-Hopf algebras as   commutative post-Hopf algebras.

Rota-Baxter operators on Lie algebras and associative algebras have important applications in various fields, such as Connes-Kreimer's algebraic approach to renormalization
of quantum  field theory \cite{CK}, the classical Yang-Baxter equation and integrable systems \cite{Bai,Ku,STS}, splitting of operads \cite{BBGN}, double Lie algebras \cite{GK} and etc. See the book \cite{Gub} for more details. Recently, the notion of Rota-Baxter operators on groups was introduced in \cite{GLS}, and further studied in \cite{BG2}. One can obtain Rota-Baxter operators of weight 1 on Lie algebras from that on Lie groups by differentiation. Then in the remarkable work  \cite{Go}, Goncharov succeeded in defining Rota-Baxter operators on cocommutative Hopf algebras such that many classical results still hold in the Hopf algebra level. In this paper, we introduce a more general notion of relative Rota-Baxter operators on  Hopf algebras containing Goncharov's Rota-Baxter operators as special cases. A cocommutative post-Hopf algebra naturally gives rise to a relative Rota-Baxter operator on its subadjacent Hopf algebra, and conversely, a relative Rota-Baxter operator   also induces a   post-Hopf algebra.

Remarkably, a relative Rota-Baxter operator  on  a  cocommutative Hopf algebra naturally gives rise to a matched pair of Hopf algebras. In particular, for a cocommutative post-Hopf algebra, the original Hopf algebra and the subadjacent Hopf algebra form a matched pair of Hopf algebras satisfying certain good properties. Based on this fact,  we construct solutions of the   Yang-Baxter equation in a Hopf algebra using post-Hopf algebras as well as relative Rota-Baxter operators, and give explicit formulas of solutions for the post-Hopf algebras coming from ordered rooted trees.   We further characterize relative Rota-Baxter operators using graphs in the smash product Hopf algebra and module structures.

The paper is organized as follows. In Section \ref{sec:pH}, first we introduce the notion of post-Hopf algebras and show that a cocommutative post-Hopf algebra gives rise to a subadjacent Hopf algebra together with a module bialgebra structure on itself. Then we show that there is a one-to-one correspondence between cocommutative post-Hopf algebras and cocommutative Hopf braces. In Section \ref{sec:rb}, we introduce the notion of relative Rota-Baxter operators and show that post-Hopf algebras are the underlying structures, and give rise to relative Rota-Baxter operators on the subadjacent Hopf algebras. In Section \ref{sec:mp}, we show that a relative Rota-Baxter operator gives rise to a matched pair of Hopf algebras. In particular, a cocommutative post-Hopf algebra gives rise to a matched pair of Hopf algebras. Consequently, one can construct solutions of the  Yang-Baxter equation using post-Hopf algebras and relative Rota-Baxter operators. In Section \ref{sec:gr}, we give some alternative characterizations of relative Rota-Baxter operators using  relative Rota-Baxter operators on the Lie algebra of primitive elements, graphs and  module bialgebra structures.

\vspace{2mm}

\noindent
{\bf Convention.}
In this paper, we fix an algebraically closed ground field $\bk$ of characteristic 0. 
For any coalgebra $(C,\Delta,\vep)$, we compress the Sweedler notation of the comultiplication $\Delta$ as $$\Delta(x)=x_1\otimes x_2$$ for simplicity.
Furthermore, for $n\geq1$ we write $$\Delta^{(n)}(x)=(\Delta\otimes\id^{\otimes (n-1)})\cdots(\Delta\otimes\id)\Delta(x)=x_1\otimes\cdots\otimes x_{n+1}.$$

Let $(H,\cdot,1,\Delta,\vep,S)$ be a  Hopf algebra.
Denote by $G(H)$ the set of group-like elements in $H$, which is a group. Denote by $P_{g,h}(H)$ the subspace of $(g,h)$-primitive elements in $H$ for $g,h\in G(H)$. Denote by $P(H)$ the subspace of primitive  elements in $H$, which is a Lie algebra.
For other basic notions of Hopf algebras, we follow the textbooks~\mcite{Mon}.

\section{Post-Hopf algebras}\label{sec:pH}

In this section, first we introduce the notion of a post-Hopf algebra, and show that a cocommutative post-Hopf algebra gives rise to a subadjacent Hopf algebra together with a module bialgebra structure on itself. A post-Hopf algebra induces a post-Lie algebra structure on the space of primitive elements and conversely, there is naturally a post-Hopf algebra structure on the universal enveloping algebra of a post-Lie algebra. Then we show that cocommutative post-Hopf algebras and cocommutative Hopf braces are equivalent. Finally, we introduce the notion of a pre-Hopf algebra which is a commutative  post-Hopf algebra.

Recall from \cite{FM,Val} that a {\bf post-Lie algebra} $(\g,[\cdot,\cdot]_\g,\rhd)$ consists of a Lie algebra $(\g,[\cdot,\cdot]_\g)$ and a binary product $\rhd:\g\otimes\g\to\g$ such that
\begin{eqnarray}
\label{Post-L-1}x\rhd[y,z]_\g&=&[x\rhd y,z]_\g+[y,x\rhd z]_\g,\\
\label{Post-L-2}([x,y]_\g+x\rhd y-y\rhd x)\rhd z&=&x\rhd(y\rhd z)-y\rhd(x\rhd z).
\end{eqnarray}
Any post-Lie algebra $(\g,[\cdot,\cdot]_\g,\rhd)$ has a {\bf subadjacent Lie algebra}
$\g_\rhd:=(\g,[\cdot,\cdot]_{\g_\rhd})$
defined by
$$[x,y]_{\g_\rhd} \coloneqq x\rhd y-y\rhd x+[x,y]_\g,\quad\forall x,y\in\g,$$
and Eqs.~\eqref{Post-L-1}-\eqref{Post-L-2} equivalently mean that the linear map $L_\rhd:\g\to\mathfrak g\mathfrak l(\g)$ defined by $L_{\rhd,x}y=x\rhd y$ is an action of the Lie algebra $(\g,[\cdot,\cdot]_{\g_\rhd})$ on $(\g,[\cdot,\cdot]_\g)$.

\subsection{Post-Hopf algebras and their basic properties}

\begin{defn}\label{defi:pH}
A {\bf post-Hopf algebra} is a pair $(H,\rhd)$, where $H$ is a Hopf algebra and $\rhd:H\otimes H\to H$ is a coalgebra homomorphism satisfying the following equalities:
\begin{eqnarray}
\label{Post-2}x\rhd (y\cdot z)&=&(x_1\rhd y)\cdot(x_2\rhd z),\\
\label{Post-4}x\rhd (y\rhd z)&=&\big(x_1\cdot(x_2\rhd y)\big)\rhd z,
\end{eqnarray}
for any $x,y,z\in H$, and the left multiplication $\alpha_\rhd:H\to \End(H)$ defined by
$$\alpha_{\rhd, x} y= x\rhd y,\quad\forall x,y\in H,$$
is convolution invertible in $\Hom(H,\End(H))$. Namely, there exists unique $\beta_\rhd:H\to\End(H)$ such that
\begin{equation}\label{Post-con}
\alpha_{\rhd,x_1}\beta_{\rhd,x_2}=\beta_{\rhd,x_1}\alpha_{\rhd,x_2}=\vep(x)\id_H,\quad\forall x\in H.
\end{equation}

A {\bf homomorphism} from a post-Hopf algebra $(H,\rhd)$ to $(H',\rhd')$ is  a Hopf algebra homomorphism $g:H\to H'$ satisfying
$$
g(x\rhd y)=g(x)\rhd'g(y),\quad \forall x,y\in H.
$$

\end{defn}

It is obvious that post-Hopf algebras and homomorphisms between post-Hopf algebras form a category, which is denoted by $\PH$. We denote by $\cocPH$ the subcategory of $\PH$ consisting of cocommutative post-Hopf algebras and homomorphisms between them.

\begin{remark}
Similar axioms in the definition of a post-Hopf algebra also appeared in the definition of $D$-algebras \cite{ML,MW} and $D$-bialgebras  \cite{MQS} with  motivations from the studies of numerical Lie group integrators and the algebraic structure on the universal enveloping algebra of a post-Lie algebra.
\end{remark}


Moreover, we have the following properties.

\begin{lemma}
Let $(H,\rhd)$ be a post-Hopf algebra. Then for all $x,y\in H$, we have
\begin{eqnarray}
\label{Post-1}x\rhd 1&=&\vep(x)1,\\
\label{Post-3}1\rhd x&=&x,\\
\label{Post-5}S(x\rhd y)&=&x\rhd S(y).
\end{eqnarray}
\end{lemma}

\begin{proof}
Since $\rhd$ is a coalgebra homomorphism, we have
\begin{eqnarray*}
x\rhd 1 &=& (x_1\rhd 1) \vep(x_2\rhd 1)
= (x_1\rhd 1)\cdot (x_2\rhd 1)\cdot S(x_3\rhd 1)\\
&\stackrel{\eqref{Post-2}}{=}& (x_1\rhd 1)\cdot S(x_2\rhd 1)= \vep(x\rhd 1)1 =\vep(x)1.
\end{eqnarray*}

By Eq.~\eqref{Post-con},  we have
$\alpha_{\rhd,1}\beta_{\rhd,1}=\beta_{\rhd,1}\alpha_{\rhd,1}=\id_H,$
which means that $\alpha_{\rhd,1}$ is a linear automorphism of $H$. On the other hand, we have
\begin{eqnarray*}
\alpha_{\rhd,1}^2x = 1\rhd (1\rhd x)
\stackrel{\eqref{Post-4}}{=}  (1\rhd 1)\rhd x
\stackrel{\eqref{Post-1}}{=}  1\rhd x = \alpha_{\rhd,1}x.
\end{eqnarray*}
Hence, $1\rhd x = \alpha_{\rhd,1}x = x$.

Finally we have
\begin{eqnarray*}
S(x\rhd y)&=&S(x_1\rhd y_1)\vep(x_2)\vep(y_2) \stackrel{\eqref{Post-1}}{=} S(x_1\rhd y_1)\cdot(x_2\rhd \vep(y_2)1)\\
&=&S(x_1\rhd y_1)\cdot(x_2\rhd (y_2\cdot S(y_3))) \stackrel{\eqref{Post-2}}{=} S(x_1\rhd y_1)\cdot(x_2\rhd y_2)\cdot(x_3\rhd S(y_3))\\
&=&\vep(x_1\rhd y_1)(x_2\rhd S(y_2))=\vep(x_1)\vep(y_1)(x_2\rhd S(y_2))=x\rhd S(y).\qedhere
\end{eqnarray*}
 \end{proof}


Now we give the main result in this section.
\begin{theorem}\label{thm:subHopf}
Let $(H,\rhd)$ be a cocommutative post-Hopf algebra. Then
$$H_\rhd\coloneqq(H,*_\rhd,1,\Delta,\vep,S_\rhd)$$
is a Hopf algebra, which is called the {\bf subadjacent Hopf algebra}, where for all $x,y\in H$,
\begin{eqnarray}
\label{post-rbb-1}x *_\rhd y&\coloneqq&x_1\cdot (x_2\rhd y),\\
\label{post-rbb-2}S_\rhd(x)&\coloneqq&\beta_{\rhd,x_1}(S(x_2)),
\end{eqnarray}

Furthermore,
$(H,\cdot,1,\Delta,\vep,S)$ is a left $H_\rhd$-module bialgebra via the action $\rhd$.

\end{theorem}
\begin{proof}
Since $\rhd$ is a coalgebra homomorphism and $H$ is cocommutative, we have
\begin{eqnarray*}
\Delta(x *_\rhd y)&=&\Delta(x_1\cdot (x_2\rhd y))\\
&=&\Delta(x_1)\cdot \Delta(x_2\rhd y)\\
&=&(x_1\otimes x_2)\cdot((x_3\rhd y_1)\otimes(x_4\rhd y_2))\\
&=&(x_1\cdot (x_3\rhd y_1))\otimes (x_2\cdot(x_4\rhd y_2))\\
&=&(x_1\cdot (x_2\rhd y_1))\otimes (x_3\cdot(x_4\rhd y_2))\\
&=&(x_1 *_\rhd y_1)\otimes(x_2 *_\rhd y_2)
\end{eqnarray*}
for all $x,y\in H$, which implies  that the comultiplication $\Delta$ is an algebra homomorphism with respect to the multiplication $*_\rhd$. Moreover, we have
$$
\vep(x *_\rhd y)=\vep(x_1\cdot (x_2\rhd y))=\vep(x_1)\vep(x_2\rhd y)=\vep(x)\vep(y),
$$
which implies  that the counit $\vep$ is also an algebra homomorphism with respect to the multiplication $*_\rhd$. Since the comultiplication $\Delta$ is an algebra homomorphism with respect to the multiplication $\cdot$, for all $x,y,z\in H$, we have
\begin{eqnarray*}
(x *_\rhd y)*_\rhd z&=& (x_1 *_\rhd y_1)\cdot((x_2 *_\rhd y_2)\rhd z)\\
&=&(x_1\cdot (x_2\rhd y_1))\cdot((x_3\cdot (x_4\rhd y_2))\rhd z)\\
&\stackrel{\eqref{Post-4}}{=}&x_1\cdot (x_2\rhd y_1)\cdot (x_3\rhd (y_2\rhd z))\\
&\stackrel{\eqref{Post-2}}{=}&x_1\cdot (x_2\rhd (y_1\cdot (y_2\rhd z)))\\
&=&x *_\rhd (y*_\rhd z),
\end{eqnarray*}
which implies that the multiplication $*_\rhd$ is associative.  For any $x\in H$, by  \eqref{Post-1} and \eqref{Post-3}, we have
\begin{eqnarray*}
x *_\rhd 1&=&x_1\cdot (x_2\rhd 1)=x_1\cdot\vep(x_2)=x,\\
1 *_\rhd x&=&1\cdot (1\rhd x)=x.
\end{eqnarray*}
Thus,   $(H,*_\rhd,1,\Delta,\vep)$ is a cocommutative bialgebra.
Since $\rhd$ is a coalgebra homomorphism and $H$ is cocommutative, we know that
$$\Delta\beta_{\rhd,x}=(\beta_{\rhd,x_1}\otimes\beta_{\rhd,x_2})\Delta,$$
and $S_\rhd$ is a coalgebra homomorphism. Also, note that
\begin{eqnarray*}
x_1*_\rhd S_\rhd(x_2) &\stackrel{\eqref{post-rbb-1}}{=}& x_1\cdot(x_2\rhd S_\rhd(x_2))\\
&\stackrel{\eqref{post-rbb-2}}{=}& x_1\cdot(\alpha_{\rhd,x_2}(\beta_{\rhd,x_3}(S(x_4))))\\
&=& x_1\cdot(\vep(x_2)S(x_3))\\
&=& \vep(x)1,
\end{eqnarray*}
and it means that
\begin{eqnarray*}
S_\rhd^2(x) &=& \vep(x_1)S_\rhd^2(x_2) = (x_1*_\rhd S_\rhd(x_2))*_\rhd S_\rhd^2(x_3)\\
&=& x_1*_\rhd (S_\rhd(x_2)*_\rhd S_\rhd(S_\rhd(x_3)))\\
&=& x_1*_\rhd \vep(S_\rhd(x_2))1 = x.\\[.5em]
S_\rhd(x_1) *_\rhd x_2 &=&  S_\rhd(x_1) *_\rhd S_\rhd^2(x_2)= \vep(S_\rhd(x))1 = \vep(x)1.
\end{eqnarray*}
Therefore, $(H,*_\rhd,1,\Delta,\vep,S_{\rhd})$ is a cocommutative Hopf algebra.

Moreover, we have
$$
(x *_\rhd y)\rhd z=(x_1\cdot (x_2\rhd y))\rhd z=x\rhd (y\rhd z).
$$
Then by \eqref{Post-2} and \eqref{Post-1}, $(H,\cdot,1)$ is a left $H_\rhd$-module algebra. Since $\rhd$ is also a coalgebra homomorphism, $(H,\cdot,1,\Delta,\vep,S)$ is a left $H_\rhd$-module  bialgebra via the action $\rhd$.
\end{proof}

\begin{remark}\begin{itemize}
\item[{\rm(i)}] The product \eqref{post-rbb-1} generalizes the Grossman-Larsson product \cite{GL,MW,OG} defined in the context of (noncommutative) polynomials of (ordered) rooted trees. The Grossman-Larsson product plays important roles in the theories of Magnus expansions \cite{CP} and Lie-Butcher series \cite{ML,MW}. 

 \item[{\rm(ii)}]  By  ~\eqref{Post-3} and \eqref{Post-4}, we have
$$\alpha_{\rhd,x_1}\alpha_{\rhd,S_\rhd(x_2)}=\alpha_{\rhd,x_1*_\rhd S_\rhd(x_2)}=\alpha_{\rhd,\vep(x)1}=\vep(x){\id},\quad{\rm i.e.}\quad  \beta_{\rhd,x}=\alpha_{\rhd,S_\rhd(x)},$$
as $\beta_\rhd$ is also the convolution inverse of $\alpha_\rhd$.
Then we can rewrite Eq.~\eqref{post-rbb-2} as
\begin{equation}\label{post-rbb-2'}
S_\rhd(x)=S_\rhd(x_1)\rhd S(x_2),\quad\forall x\in H.
\end{equation}
\end{itemize}
\end{remark}

\begin{exam}
Any Hopf algebra $H$ has at least the following {\bf trivial} post-Hopf algebra structure,
\begin{eqnarray*}
x\rhd y &=& \vep(x)y,\quad\forall x,y\in H.
\end{eqnarray*}
\end{exam}

In the sequel, we study the relation between post-Hopf algebras and post-Lie algebras.

\begin{theorem}
Let $(H,\rhd)$ be a  post-Hopf algebra. Then its subspace $P(H)$ of primitive elements is a post-Lie algebra.
\end{theorem}

\begin{proof}
Since $\rhd$ is a coalgebra homomorphism, for all $x,y\in P(H)$, we have
\begin{eqnarray*}
\Delta(x\rhd y)&=&(x_1\rhd y_1)\otimes (x_2\rhd y_2)\\
               &=&(1\rhd 1)\otimes (x\rhd y)+(1\rhd y)\otimes (x\rhd 1)+(x\rhd 1)\otimes (1\rhd y)+(x\rhd y)\otimes (1\rhd 1)\\
               &\stackrel{\eqref{Post-1},\,\eqref{Post-3}}{=}&1\otimes (x\rhd y)+(x\rhd y)\otimes 1.
\end{eqnarray*}
Thus, we obtain a linear map $\rhd:P(H)\otimes P(H)\to P(H)$. By \eqref{Post-2}, for all $x,y\in P(H)$, we have
\begin{eqnarray*}
x\rhd (y\cdot z)&=&(1\rhd y)\cdot(x\rhd z)+(x\rhd y)\cdot(1\rhd z)\\
                &\stackrel{\eqref{Post-3}}{=}&y\cdot(x\rhd z)+(x\rhd y)\cdot z.
\end{eqnarray*}
Thus, we have
\begin{eqnarray*}
x\rhd [y,z]&=&x\rhd (y\cdot z)-x\rhd (z\cdot y)\\
           &=&y\cdot(x\rhd z)+(x\rhd y)\cdot z-z\cdot(x\rhd y)-(x\rhd z)\cdot y\\
           &=&[x\rhd y,z]+[y,x\rhd z].
\end{eqnarray*}
By \eqref{Post-4}, we have
\begin{eqnarray*}
x\rhd (y\rhd z)&=&\big(1\cdot(x\rhd y)\big)\rhd z+\big(x\cdot(1\rhd y)\big)\rhd z=(x\rhd y)\rhd z+\big(x\cdot y\big)\rhd z.
\end{eqnarray*}
Thus, we have
\begin{eqnarray*}
[x,y]\rhd z&=&(x\cdot y)\rhd z-(y\cdot x)\rhd z\\
           &=&x\rhd (y\rhd z)-(x\rhd y)\rhd z-y\rhd (x\rhd z)+(y\rhd x)\rhd z.
\end{eqnarray*}
Therefore,  $(P(H),[\cdot,\cdot],\rhd)$ is a post-Lie algebra.
\end{proof}

In \cite{ELM,OG} the authors studied the universal enveloping algebra of a pre-Lie algebra and also of a post-Lie algebra. By \cite[Proposition~3.1, Theorem~ 3.4]{ELM}, the binary product  $\rhd$ in a post-Lie algebra $(\g,[\cdot,\cdot]_\g,\rhd)$ can be extended to its universal enveloping algebra and induces a subadjacent Hopf algebra structure isomorphic to the universal enveloping algebra $U(\g_\rhd)$ of the subadjacent Lie algebra $\g_\rhd$.

We summarize their result in the setting of post-Hopf algebras as follows.
We do not claim any originality (see   \cite{ELM,OG}
for details).

\begin{theorem}\label{th:post-uea}
Let $(\g,[\cdot,\cdot]_\g,\rhd)$ be a post-Lie algebra with its subadjacent Lie algebra $\g_\rhd$. Then  $(U(\g),{\bar \rhd})$ is a post-Hopf algebra, where ${\bar \rhd}$ is the extension of $\rhd$  determined by
$$1{\bar \rhd} u = u,\quad x_1\cdots x_r{\bar \rhd} u = x_1{\bar \rhd}(x_2\cdots x_r{\bar \rhd} u)-(x_1{\bar \rhd} x_2\cdots x_r){\bar \rhd} u$$
for all $x_1,\dots,x_r\in \g$ and $u\in U(\g)$ with $r\geq1$.

Moreover, the subadjacent Hopf algebra $U(\g)_{\bar \rhd}$ is isomorphic to the universal enveloping algebra $U(\g_\rhd)$ of the subadjacent Lie algebra $\g_\rhd$.
\end{theorem}

In a recent work \cite{Foissy}, Foissy extended any magma operation  on a vector space $V$, i.e. an arbitrary  bilinear map $\cast:V \otimes V\to V$,
to the coshuffle Hopf algebra $(TV,\cdot,\Delta^{\co})$ as follows:
\begin{eqnarray*}\label{post-brace-algebra-1}
1\rhd a&=&a,\\
x\rhd a&=&x\cast a,\\
(x\otimes x_1) \rhd a&=&x\cast (x_1\rhd a)-(x\cast x_1)\rhd a,\\
&\vdots&\\
(x\otimes x_1\otimes\cdots \otimes x_n) \rhd a&=&x\cast \big((x_1\otimes\cdots \otimes x_n)\rhd a\big)\\
&&-\sum_{i=1}^{n}\big(x_1\otimes\cdots \otimes x_{i-1}\otimes(x\cast x_i) \otimes x_{i+1}\otimes\cdots\otimes x_n\big)\rhd a,
\end{eqnarray*}
and
\begin{eqnarray*}\label{post-brace-algebra-2}
1\rhd 1&=&1,\\
x\rhd 1&=&0,\\
X \rhd (a_1\otimes\cdots\otimes a_m)&=&(X_1\rhd a_1)\otimes\cdots\otimes(X_m\rhd a_m),
\end{eqnarray*}
where $x,x_1,\dots,x_n,a,a_1,\dots,a_m\in V$ and $X\in TV,~{\Delta^{\co}}^{(m-1)}X=X_1\otimes\cdots\otimes X_m$.

According to the discussion in \cite{Foissy}, it is straightforward to obtain the following result.

\begin{theorem}\label{magma-to-post}
Let $(V,\cast)$ be a magma algebra. Then $(TV,\cdot,\Delta^{\co},\rhd)$ is a post-Hopf algebra.
\end{theorem}

\begin{exam}\label{free-post-lie}
Let $\huaO$ be the set of isomorphism classes of ordered rooted trees, which is denoted  by
\[
	\huaO= \Big\{\begin{array}{c}
		\scalebox{0.6}{\ab}, \scalebox{0.6}{\aabb},
		\scalebox{0.6}{\aababb}, \scalebox{0.6}{\aaabbb},\scalebox{0.6}{\aabababb},
		\scalebox{0.6}{\aaabbabb},
		\scalebox{0.6}{\aabaabbb}, \scalebox{0.6}{\aaababbb}, \scalebox{0.6}{\aaaabbbb},\scalebox{0.6}{\aababababb},\scalebox{0.6}{\aaabbababb},\scalebox{0.6}{\aabaabbabb},\scalebox{0.6}{\aababaabbb},\scalebox{0.6}{\aaababbabb},\scalebox{0.6}{\aabaababbb},\scalebox{0.6}{\aaabbaabbb},\scalebox{0.6}{\aaabababbb},\scalebox{0.6}{\aaaabbbabb},\scalebox{0.6}{\aabaaabbbb},\scalebox{0.6}{\aaaababbbb},\scalebox{0.6}{\aaaabbabbb},\scalebox{0.6}{\aaabaabbbb},\scalebox{0.6}{\aaaaabbbbb},\ldots
 			\end{array}
			\Big\}.
\]
Let $\bk\{\huaO\}$ be the free $\bk$-vector space generated by  $\huaO$. The {\bf left grafting operator} $\curvearrowright:\bk\{\huaO\}\otimes \bk\{\huaO\}\to \bk\{\huaO\}$ is defined by
\begin{eqnarray*}
\tau\curvearrowright \omega=\sum_{s\in {\rm Nodes}(\omega)}\tau\circ_{s}\omega,\quad \forall \tau,\omega\in \huaO,
\end{eqnarray*}
where $\tau\circ_{s}\omega$ is the ordered rooted tree resulting from attaching the root of $\tau$ to  the node $s$ of the tree $\omega$ from the left. For example, we have
\begin{eqnarray*}
\scalebox{0.6}{\aabb}\curvearrowright \scalebox{0.6}{\aabababb}=\scalebox{0.6}{\aaabbabababb}+\scalebox{0.6}{\aaaabbbababb}+\scalebox{0.6}{\aabaaabbbabb}+\scalebox{0.6}{\aababaaabbbb},\qquad
\scalebox{0.6}{\aababb}\curvearrowright\scalebox{0.6}{\aaabbabb}=\scalebox{0.6}{\aaababbaabbabb}+\scalebox{0.6}{\aaaababbabbabb}+\scalebox{0.6}{\aaaaababbbbabb}+\scalebox{0.6}{\aaabbaaababbbb}.
\end{eqnarray*}
It is obvious that $(\bk\{\huaO\},\curvearrowright)$ is a magma algebra. By Theorem \ref{magma-to-post}, $(T\bk\{\huaO\},\cdot,\Delta^{\co},\rhd)$ is a post-Hopf algebra,
where the underlying coshuffle Hopf algebra $(T\bk\{\huaO\},\cdot,\Delta^{\co})$ has the linear basis consisting of all ordered rooted forests and its antipode $S$ is given by
$$S(\tau_1 \tau_2 \cdots \tau_m)=(-1)^m \tau_m \tau_{m-1} \cdots \tau_1,\quad\forall \tau_1,\tau_2,\dots,\tau_m\in\huaO.$$
Moreover, it is the universal enveloping algebra of the free post-Lie algebra on one generator $\{\scalebox{0.6}{\ab}\}$. See \cite{Foissy,ML} for more details about free post-Lie algebras and their universal enveloping algebras.

Let $B^+:T\bk\{\huaO\}\to \bk\{\huaO\}$ be the linear map
  producing an ordered tree $\tau$ from any ordered rooted forest $\tau_1\cdots\tau_m$ by grafting the $m$ trees $\tau_1,\dots,\tau_m$ on a new root $\scalebox{0.6}{\ab}$ in order.
For example, we have
\begin{eqnarray*}
B^+(\scalebox{0.6}{\aabb\,\,\,\aababb})=\scalebox{0.6}{\aaabbaababbb}.
\end{eqnarray*}
Let $B^-:\bk\{\huaO\}\to T\bk\{\huaO\}$  be the linear map
 producing an ordered forest from any ordered rooted tree $\tau$ by removing its root.
For example, we have
\begin{eqnarray*}
B^-(\scalebox{0.6}{\aababaaabbbb})=\scalebox{0.6}{\ab\,\,\ab\,\,\aaabbb}.
\end{eqnarray*}
Moreover, the operation $B^-$ extends to $T\bk\{\huaO\}$ by
\begin{eqnarray*}
B^-(\tau_1\cdots\tau_m)=B^-(\tau_1)\cdots B^-(\tau_m),\quad\forall \tau_1,\dots,\tau_m\in\huaO.
\end{eqnarray*}

Note that the subadjacent Hopf algebra $(T\bk\{\huaO\},{*_\rhd},\Delta^{\co},S_\rhd)$ is isomorphic to the Grossman-Larson Hopf algebra of ordered rooted trees defined in \cite{GL}. Using the left grafting operation, the multiplication $*_\rhd$ is given by
\begin{eqnarray*}
\huaX *_\rhd \huaY&=&B^-(\huaX\rhd B^+(\huaY))
\end{eqnarray*}
for all ordered rooted forests $\huaX,\huaY$, and the antipode $S_\rhd$ can be recursively defined by
$$S_\rhd(1) = 1,\quad
S_\rhd(\huaX) \stackrel{\eqref{Post-3},\,\eqref{post-rbb-2'}}{=}
S(\huaX)+(\Id-\mu\vep)(S_\rhd(\huaX_1))\rhd S(\huaX_2),$$
where $\mu$ is the unit map and $\vep$ is the counit map.
\end{exam}

\begin{exam}
We classify all post-Hopf algebra structures on the smallest noncommutative and non-cocommutative Hopf algebra, namely, Sweedler's 4-dimensional Hopf algebra
\[H_4=\bk\langle 1,g,x,gx\,|\, g^2=1,x^2=0,gx=-xg\rangle,\]
with its coalgebra structure and its antipode given by
\[\Delta(g)=g\otimes g,\quad \Delta(x)=x\otimes 1+g\otimes x,\quad \vep(g)=1,\quad \vep(x)=0,\quad S(g)=g,\quad S(x)=-gx.\]
Further $G(H_4)=\{1,g\}$, $P_{1,g}(H_4)=\bk x$ and $P_{g,1}(H_4)=\bk gx$.

Let $(H_4,\rhd)$ be a post-Hopf algebra structure on $H_4$. Then
\begin{eqnarray*}
\Delta(g\rhd g) &=& (g\rhd g)\otimes(g\rhd g),\\
\Delta(g\rhd x) &=& (g\rhd x)\otimes(g\rhd 1)+(g\rhd g)\otimes(g\rhd x) \stackrel{\eqref{Post-1}}{=}   (g\rhd x)\otimes 1 + (g\rhd g)\otimes(g\rhd x).
\end{eqnarray*}
Namely, $g\rhd g\in G(H_4)$ and $g\rhd x\in P_{1,g\rhd g}(H_4)$. Since $g\in G(H_4)$ implies that $\alpha_{\rhd,g}$ is invertible by Eq.~\eqref{Post-con}, we know that $g\rhd g=g$ and $g\rhd x\in P_{1,g}(H_4)\setminus\{0\}$. Also,
\begin{eqnarray*}
g\rhd(g\rhd x) \stackrel{\eqref{Post-1}}{=} (g(g\rhd g)) \rhd x = g^2\rhd x= 1\rhd x \stackrel{\eqref{Post-3}}{=} x.
\end{eqnarray*}
Therefore, $g\rhd x=x$ or $-x$. On the other hand,
\begin{eqnarray*}
\Delta(x\rhd g) &=& (x\rhd g)\otimes (1\rhd g) + (g\rhd g)\otimes (x\rhd g) \stackrel{\eqref{Post-3}}{=}  (x\rhd g)\otimes g + g\otimes (x\rhd g).
\end{eqnarray*}
Then $x\rhd g\in P_{g,g}(H_4)$, and thus $x\rhd g=0$. So
\begin{eqnarray*}
\Delta(x\rhd x) &=& (x\rhd x)\otimes (1\rhd 1) + (g\rhd x)\otimes (x\rhd 1)
+ (x\rhd g)\otimes (1\rhd x) + (g\rhd g)\otimes (x\rhd x)\\ &\stackrel{\eqref{Post-1},\eqref{Post-3}}{=}&
(x\rhd x)\otimes 1 + g\otimes (x\rhd x).
\end{eqnarray*}
That is, $x\rhd x\in P_{1,g}(H_4)$, and we can set $x\rhd x=a x$ for some $a\in \bk$.
Then
\begin{eqnarray*}
a(g\rhd x) &=& x\rhd(g\rhd x) \stackrel{\eqref{Post-4}}{=} (x(1\rhd g)+g(x\rhd g)) \rhd x = xg\rhd x = -gx\rhd x\\
 &=& g\rhd(x\rhd x) \stackrel{\eqref{Post-4}}{=} (g(g\rhd x)) \rhd x.
\end{eqnarray*}
It implies that $g\rhd x=-x$ unless $a=0$.

In summary, one can easily check that there is
the post-Hopf algebra structure $(H_4,\rhd_a)$ for any $a\in\bk$ illustrated as below, 
such that $\alpha_{\rhd_a}$ has the convolution inverse $\alpha_{\rhd_{-a}}$.
\begin{center}
\begin{tabular}{|c|c|c|c|c|}
\hline
$\rhd_a$ & $1$ & $g$ & $x$ & $gx$\\
\hline
$1$ & $1$ & $g$ & $x$ & $gx$\\
\hline
$g$ & $1$ & $g$ & $-x$ & $-gx$\\
\hline
$x$ & $0$ & $0$ & $ax$ & $agx$\\
\hline
$gx$ & $0$ & $0$ & $ax$ & $agx$\\
\hline
\end{tabular}
\end{center}
Moreover, if $a\neq 0$, there is the post-Hopf algebra isomorphism from $(H_4,\rhd_a)$ to $(H_4,\rhd_1)$ mapping $g$ to $g$ and $x$ to $ax$. Hence, the Sweedler 4-dimensional Hopf algebra has three non-isomorphic post-Hopf algebra structures
 $(H_4,\vep\otimes\id)$, $(H_4,\rhd_0)$ and $(H_4,\rhd_1)$.

\delete{
It is interesting to see that the {\it non-cocommutative} post-Hopf algebra $(H_4,\rhd_a)$ also has the subadjacent algebra $H_{4,\rhd_a}$ with its multiplication $*_{\rhd_a}$ defined as in Eq.~\eqref{post-rbb-1}, which becomes commutative as shown below but not compatible with the original comultiplication $\Delta$.
\begin{center}
\begin{tabular}{|c|c|c|c|c|}
\hline
$*_{\rhd_a}$ & $1$ & $g$ & $x$ & $gx$\\
\hline
$1$ & $1$ & $g$ & $x$ & $gx$\\
\hline
$g$ & $g$ & $1$ & $-gx$ & $-x$\\
\hline
$x$ & $x$ & $-gx$ & $agx$ & $ax$\\
\hline
$gx$ & $gx$ & $-x$ & $ax$ & $agx$\\
\hline
\end{tabular}
\end{center}
}
\end{exam}

\subsection{Post-Hopf algebras and Hopf braces}
In this subsection, we establish the relation between Hopf braces and post-Hopf algebras.

In \cite{AGV}, Angiono, Galindo and Vendramin introduced the notion of Hopf braces, generalizing Rump's braces and their nonabelian generalizations, skew left braces \cite{GV,JKVV,Rump1,Ru}, which are stemmed from group theory.
All these algebraic objects have deep relations with the  Yang-Baxter equation.

\begin{defn}[{\cite[Definition~1.1]{AGV}}]
Let $(A,\Delta,\vep)$ be a coalgebra. A {\bf Hopf brace} over $A$ consist of two Hopf algebra structures $(A,\cdot,1,\Delta,\vep,S)$ and $(A,\circ,1,\Delta,\vep,\huaS)$ satisfying the following compatibility condition,
\begin{equation}\label{eq:brace}
a\circ(b\cdot c)=(a_1\circ b)\cdot S(a_2)\cdot(a_3\circ c),\quad\forall a,b,c\in A.
\end{equation}
\end{defn}
We will simply denote a Hopf brace by $(A,\cdot,\circ)$.

\begin{theorem}
Let $(H,\rhd)$ be a cocommutative post-Hopf algebra. Then the Hopf algebra $(H,\cdot,1,\Delta,\vep,S)$  and  the subadjacent Hopf algebra $(H,*_\rhd,1,\Delta,\vep,S_\rhd)$ form a Hopf brace. Conversely, any cocommutative Hopf brace $(H,\cdot,\circ)$ gives a post-Hopf algebra  $(H,\rhd)$ with $\rhd$ defined by $$x\rhd y=S(x_1)\cdot(x_2\circ y),\quad \forall x,y\in H.$$
\end{theorem}
\begin{proof}
Let $(H,\rhd)$ be a cocommutative post-Hopf algebra. We only need to show that the multiplications  $\cdot$ and $*_\rhd$ satisfy the compatibility  condition \eqref{eq:brace}, which follows from
\begin{align*}
x*_\rhd(y\cdot z)&=x_1\cdot (x_2\rhd(y\cdot z))\\
&=x_1\cdot((x_2\rhd y)\cdot (x_3\rhd z))\\
&=x_1\cdot(x_2\rhd y)\cdot S(x_3)\cdot x_4\cdot (x_5\rhd z)\\
&=(x_1*_\rhd y)\cdot S(x_2)\cdot (x_3*_\rhd z),
\end{align*}
for any $x,y,z\in H$.

Conversely, it is straightforward but tedious to check that a cocommutative Hopf brace $(H,\cdot,\circ)$ induces a post-Hopf algebra  $(H,\rhd)$.
\end{proof}


\subsection{Pre-Hopf algebras}

A post-Lie algebra $(\g,[\cdot,\cdot]_\g,\triangleright )$ reduces to a pre-Lie algebra if the Lie bracket $  [\cdot,\cdot]_\g$ is abelian. More precisely,  a {\bf pre-Lie algebra} $(\g, \rhd)$ is a vector space $\g$ equipped with a binary product $\rhd:\g\otimes\g\to\g$ such that
\begin{eqnarray}
(x\rhd y)\rhd z-x\rhd (y\rhd z)=(y\rhd x)\rhd z-y\rhd (x\rhd z),\quad \forall x,y,z\in \g.
\end{eqnarray}
From this perspective, we introduce the notion of   pre-Hopf algebras as   special post-Hopf algebras.

\begin{defn}
  A post-Hopf algebra $(H,\rhd)$ is called a {\bf pre-Hopf algebra} if $H$ is a commutative Hopf algebra.
\end{defn}

The above properties for post-Hopf algebras are still valid for pre-Hopf algebras.

\begin{coro}
 Let $(H,\rhd)$ be a cocommutative pre-Hopf algebra. Then
$$H_\rhd\coloneqq(H,*_\rhd,1,\Delta,\vep,S_\rhd)$$
is a Hopf algebra, which is called the {\bf subadjacent Hopf algebra}, where the multiplication $*_\rhd$ and the antipode $S_\rhd$ are given by \eqref{post-rbb-1} and \eqref{post-rbb-2} respectively.

 Moreover, $H$ is a left $H_\rhd$-module bialgebra via the action $\rhd$.
\end{coro}

\begin{coro}
Let $(H,\rhd)$ be a  pre-Hopf algebra. Then its subspace $P(H)$ of primitive elements is a pre-Lie algebra.
\end{coro}

Recall that a pre-Lie algebra $(\g, \rhd)$ also gives rise to a subadjacent Lie algebra $\g_\rhd$ in which the Lie bracket is defined by
$$
[x,y]_{\g_\rhd}=x\rhd y-y\rhd x,\quad \forall x,y\in \g.
$$

\begin{coro}
Let $(\g, \rhd)$ be a pre-Lie algebra with its subadjacent Lie algebra $\g_\rhd$. Then the product $\rhd$  can be extended to the one ${\bar \rhd}$ on the symmetric algebra $\Sym(\g)$, making it a pre-Hopf algebra.
Moreover, the subadjacent Hopf algebra $\Sym(\g)_{\bar \rhd}$ is isomorphic to the universal enveloping algebra $U(\g_\rhd)$ of the subadjacent Lie algebra $\g_\rhd$.
\end{coro}

\begin{exam}\label{free-pre-lie}
Let $\huaT$ be the set of isomorphism classes of rooted trees, which is denoted  by
\[
	\huaT= \Big\{\begin{array}{c}
		\scalebox{0.6}{\ab}, \scalebox{0.6}{\aabb},
		\scalebox{0.6}{\aababb}, \scalebox{0.6}{\aaabbb},\scalebox{0.6}{\aabababb},
		\scalebox{0.6}{\aaabbabb}=
		\scalebox{0.6}{\aabaabbb}, \scalebox{0.6}{\aaababbb}, \scalebox{0.6}{\aaaabbbb},\scalebox{0.6}{\aababababb},\scalebox{0.6}{\aaabbababb}=\scalebox{0.6}{\aabaabbabb}=\scalebox{0.6}{\aababaabbb},\scalebox{0.6}{\aaababbabb}=\scalebox{0.6}{\aabaababbb},\scalebox{0.6}{\aaabbaabbb},\scalebox{0.6}{\aaabababbb},\scalebox{0.6}{\aaaabbbabb}=\scalebox{0.6}{\aabaaabbbb},\scalebox{0.6}{\aaaababbbb},\scalebox{0.6}{\aaaabbabbb}=\scalebox{0.6}{\aaabaabbbb},\scalebox{0.6}{\aaaaabbbbb},\ldots
 			\end{array}
			\Big\}.
\]
Let $\bk\{\huaT\}$ be the free $\bk$-vector space generated by  $\huaT$. The grafting operator $\curvearrowright:\bk\{\huaT\}\otimes \bk\{\huaT\}\to \bk\{\huaT\}$ is defined by
\begin{eqnarray*}
\tau\curvearrowright \omega=\sum_{s\in{\rm Nodes}(\omega)}\tau\circ_{s}\omega,\quad \forall \tau,\omega\in \huaT,
\end{eqnarray*}
where $\tau\circ_{s}\omega$ is the rooted tree resulting from attaching the root of $\tau$ to the node $s$ of the tree $\omega$. For example, we have
\begin{eqnarray*}
\scalebox{0.6}{\aabb}\curvearrowright \scalebox{0.6}{\aabababb}=\scalebox{0.6}{\aaabbabababb}+\scalebox{0.6}{\aaaabbbababb}+\scalebox{0.6}{\aabaaabbbabb}+\scalebox{0.6}{\aababaaabbbb}=\scalebox{0.6}{\aaabbabababb}+{3}\scalebox{0.6}{\aaaabbbababb},\qquad
\scalebox{0.6}{\aababb}\curvearrowright\scalebox{0.6}{\aaabbabb}=\scalebox{0.6}{\aaababbaabbabb}+\scalebox{0.6}{\aaaababbabbabb}+\scalebox{0.6}{\aaaaababbbbabb}+\scalebox{0.6}{\aaabbaaababbbb}.
\end{eqnarray*}
Moreover, Chapoton and Livernet \cite{CL} have shown that $(\bk\{\huaT\},\curvearrowright)$ is the free pre-Lie algebra generated by $\{\scalebox{0.6}{\ab}\}$.
By Theorem \ref{magma-to-post}, we deduce that  $(T\bk\{\huaT\},\cdot,\Delta^{\co},\rhd)$ is a post-Hopf algebra. Since $(\bk\{\huaT\},\curvearrowright)$ is a pre-Lie algebra, the post-Hopf algebra structure reduces to the symmetric algebra $S\bk\{\huaT\}$. Thus, we deduce that $(S\bk\{\huaT\},\cdot,\Delta^{\co},\rhd)$ is a pre-Hopf algebra. Furthermore, it is the universal enveloping algebra of the free pre-Lie algebra $(\bk\{\huaT\},\curvearrowright)$, and its subadjacent Hopf algebra
 $(S\bk\{\huaT\},*_\rhd,\Delta^{\co},S_\rhd)$ is dual to the Connes-Kreimer Hopf algebra of rooted trees.

\end{exam}

\emptycomment{
\begin{exam}
Consider the $2$-dimensional complex pre-Lie algebra $(A_{\hbar},\cdot)$ defined with respect to a basis $\{e_1,e_2\}$  by
\begin{eqnarray*}
e_1\cdot e_1=e_1\cdot e_2=0,\quad e_2\cdot e_1=-e_1,\quad e_2\cdot e_2=\hbar e_2,\quad \hbar\in \mathbb C.
\end{eqnarray*}
Note that $\hbar=-1$, $(A_{-1},\cdot)$ is an associative algebra.

\end{exam}
}

\section{Relative Rota-Baxter operators on  Hopf algebras}\label{sec:rb}

In this section, first we recall relative Rota-Baxter operators on Lie algebras and groups, and Rota-Baxter operators on cocommutative Hopf algebras. Then we introduce a more general notion of relative Rota-Baxter operators of weight 1 on  cocommutative  Hopf algebras with respect to module bialgebras. We establish the relation between the category of  relative Rota-Baxter operators of weight 1 on cocommutative  Hopf algebras and  the category of post-Hopf algebras.

Let  $\phi: \g\rightarrow\Der(\h)$ be an action of a Lie
algebra $(\g, [\cdot,\cdot]_{\g})$ on a Lie algebra $ (\h,
[\cdot,\cdot]_{\h})$. A linear map $T: \h\rightarrow\g$ is called a {\bf
  relative Rota-Baxter operator (of weight $1$)}  on $\g$ with respect
to $(\h;\phi)$  if
\begin{equation} \label{eq:B}
[T(u), T(v)]_\g=T\Big(\phi(T(u))v-\phi(T(v))u+[u,v]_\h\Big),\quad \forall u, v\in\h.
\end{equation}

\medskip
Let  $\Phi: \huaH\rightarrow \Aut(\huaK)$ be an action of a group $\huaH$ on a group $\huaK$. A map $T: \huaK\rightarrow \huaH$ is called a {\bf relative Rota-Baxter operator (of weight $1$)} if
\begin{equation}\label{rRBg}
T(h)\cdot_{\huaH}T(k)=T(h\cdot_{\huaK}\Phi(T(h))k), \quad \forall h, k\in \huaK.
\end{equation}

Given any Hopf algebra $(H,\Delta,\varepsilon,S)$, define the     adjoint action  of $H$ on itself by
$\ad_x y=x_1yS(x_2).$
A {\bf Rota-Baxter operator
(of weight 1)} on a   cocommutative  Hopf algebra $H$ was defined by Goncharov in \cite{Go}, which is a coalgebra homomorphism $B$ satisfying
\begin{equation}\mlabel{eq:rb-Hopf}
B(x)B(y)=B\left(x_1\ad_{B(x_2)}y\right)=B\left(x_1 B(x_2)y S(B(x_3))\right),\quad \forall\,x,y\in H.
\end{equation}

In the sequel, all the (relative) Rota-Baxter operators under consideration are of weight 1, so we will not emphasize it anymore.

Now we generalize the above adjoint action to arbitrary actions and introduce the notion of relative Rota-Baxter operators on   Hopf algebras.

\begin{defn}
Let $H$ and $K$ be two Hopf algebras such that $K$ is a    left $H$-module bialgebra via an action $\rightharpoonup$. A coalgebra homomorphism $T:K\to H$ is called a {\bf relative Rota-Baxter operator}   with respect to the left $H$-module bialgebra $(K,\rightharpoonup)$ if the following
equality holds:
\begin{eqnarray}\mlabel{eq:rrb-Hopf}
T(a)T(b)=T\big(a_1(T(a_2)\rightharpoonup b)\big),\quad \forall  a,b\in K.
\end{eqnarray}
A {\bf homomorphism} between two relative Rota-Baxter operators
$T:K\to H$ and $T':K'\to H'$ is a pair of Hopf algebra homomorphisms $f:H\to H'$ and $g:K\to K'$ such that
\begin{eqnarray}\mlabel{eq:rrb-homo}
fT=T'g,\quad g(x\rightharpoonup a)=f(x)\rightharpoonup g(a) ,\quad \forall\,x\in H,\,a\in K.
\end{eqnarray}
\end{defn}

It is obvious that relative Rota-Baxter operators on Hopf algebras and  homomorphisms between them form a category, which is denoted by  $\rRB$. We denote by $\cocrRB$ the subcategory of $\rRB$ consisting of relative Rota-Baxter operators    with respect to cocommutative  left  module bialgebras and homomorphisms between them.

A cocommutative post-Hopf algebra naturally gives rise to a relative Rota-Baxter operator.

\begin{prop}\label{prop:PHtoRB}
Let $(H,\rhd)$ be a cocommutative post-Hopf algebra and $H_\rhd$ the subadjacent Hopf algebra. Then
the identity map $\id_H:H\to H_\rhd$ is a relative Rota-Baxter operator   with respect to the left $H_\rhd$-module bialgebra $(H,\rhd)$.

Moreover, if $g:H\to H'$ is a post-Hopf algebra homomorphism from $(H,\rhd)$ to $(H',\rhd')$, then $(g,g)$ is a homomorphism from the relative Rota-Baxter operator $\id_H:H\to H_\rhd$ to $\id_{H'}:H'\to H'_{\rhd'}$. Consequently, we obtain a functor $\Upsilon:\cocPH\to \cocrRB$ from the category of cocommutative post-Hopf algebras to the category of relative Rota-Baxter operators with respect to cocommutative  left  module bialgebras.
\end{prop}
\begin{proof}
For any $x,y,z\in H$,  we have
$$\id_H(x)*_\rhd\id_H(y)=x*_\rhd y=x_1\cdot (x_2\rhd y)=\id_H(x_1\cdot (\id_H(x_2)\rhd y)),$$
so $\id_H:H\to H_\rhd$ is a relative Rota-Baxter operator   with respect to the left $H_\rhd$-module bialgebra $(H,\rhd)$.

Let  $g:H\to H'$ be a post-Hopf algebra homomorphism from $(H,\rhd)$ to $(H',\rhd')$. Then $(g,g)$ obviously satisfy Eq.~\eqref{eq:rrb-homo}. Since $g$ is a coalgebra homomorphism and
\begin{eqnarray*}
g(x *_\rhd y)&=&g(x_1\cdot (x_2\rhd y))=g(x_1)\cdot' (g(x_2)\rhd' g(y))=g(x) *_{\rhd'} g( y),
\end{eqnarray*}
we deduce  that $g$ is a homomorphism from the Hopf algebra $H_\rhd$ to $H'_{\rhd'}$. Therefore, $(g,g)$ is a homomorphism from the relative Rota-Baxter operator $\id_H:H\to H_\rhd$ to $\id_{H'}:H'\to H'_{\rhd'}$. It is straightforward to check that this is indeed a functor.
\end{proof}

It is well-known that a relative Rota-Baxter operator $T:\h\to \g$ on a Lie algebra $\g$ with respect to an action $(\h;\phi)$ endows $\h$ with the following post-Lie algebra structure $\rhd_T$,
\begin{equation}\label{eq:ind-post-L}
u\rhd_T v = \phi(T(u))v,\quad\forall u,v\in \h.
\end{equation}

\begin{theorem}\label{th:post-action}
Let $T:K\to H$ be a relative Rota-Baxter operator   with respect to a left $H$-module bialgebra $(K,\rightharpoonup)$. Then there exists a post-Hopf algebra structure $\rhd_T:K\otimes K\to K$ on $K$ given by
\begin{eqnarray}\label{eq:post-action}
a\rhd_T b=T(a)\rightharpoonup b.
\end{eqnarray}

 Let $T:K\to H$ and $T':K'\to H'$ be two relative Rota-Baxter operators and $(f,g)$   a homomorphism between them. Then $g$ is a  homomorphism from the post-Hopf algebra $(K,\rhd_T)$ to  $(K',\rhd_{T'})$. Consequently, we obtain a functor $\Xi:\rRB\to \PH$ from the category of relative Rota-Baxter operators on Hopf algebras to the category of post-Hopf algebras.

Moreover, the functor $\Xi|_{\cocrRB}$ is right adjoint to the functor $\Upsilon$ given in Proposition \ref{prop:PHtoRB}.
\end{theorem}
\begin{proof}
  Since $T$ is a coalgebra homomorphism and $\rightharpoonup$ is the left module bialgebra action,  we have
\begin{eqnarray*}
\Delta_K(a\rhd_T b)&=&\Delta_K(T(a)\rightharpoonup b)=(T(a_1)\rightharpoonup b_1)\otimes(T(a_2)\rightharpoonup b_2)= (a_1\rhd_T b_1)\otimes(a_2\rhd_T b_2),\\
\vep_K(a\rhd_T b)&=&\vep_K(T(a)\rightharpoonup b)=\vep_H(T(a))\vep_K(b)=\vep_K(a)\vep_K(b),
\end{eqnarray*}
which implies that $\rhd_T$ is a coalgebra homomorphism. Similarly, we have
\begin{eqnarray*}
a\rhd_T(bc)& =&T(a)\rightharpoonup (bc)=(T(a_1)\rightharpoonup b)(T(a_2)\rightharpoonup c)=(a_1\rhd_T b)(a_2\rhd_T c).
\end{eqnarray*}
Then by \eqref{eq:rrb-Hopf}, we obtain
\begin{eqnarray*}
(a_1(a_2 \rhd_T b))\rhd_T c
&=&T(a_1(T(a_2)\rightharpoonup b))\rightharpoonup c=(T(a)T(b))\rightharpoonup c\\
&=&T(a)\rightharpoonup(T(b)\rightharpoonup c)=a\rhd_T(b\rhd_T c).
\end{eqnarray*}
Define linear map $S_T:K\to K$ by  \begin{eqnarray}
\label{rrb-des2}S_T(a)&=&S_H(T(a_1))\rightharpoonup S_K(a_2).
\end{eqnarray}
Then for all $a\in K$, we have
\begin{eqnarray}
\nonumber T(S_T(a)) &=& \vep_H(T(a_1))T(S_T(a_2))\\
\nonumber&=& S_H(T(a_1))T(a_2)T(S_T(a_3))\\
\nonumber&\stackrel{\eqref{eq:rrb-Hopf}}{=}& S_H(T(a_1))T(a_2(T(a_3)\rightharpoonup S_T(a_4)))\\
\nonumber&=& S_H(T(a_1))T(a_2(T(a_3)\rightharpoonup (S_H(T(a_4))\rightharpoonup S_K(a_5))))\\
\nonumber&=& S_H(T(a_1))T(a_2(T(a_3)S_H(T(a_4))\rightharpoonup S_K(a_5)))\\
\nonumber&=& S_H(T(a_1))T(a_2S_K(a_3))\\
\nonumber&=& S_H(T(a_1))T(\vep_K(a_2)1)\\
\label{eq:rrb-anti}&=&S_H(T(a)).
\end{eqnarray}
For all $a\in K$, define $\beta_{\rhd_T,a}\in \End(K)$ by $\beta_{\rhd_T,a}\coloneqq \alpha_{\rhd_T,S_T(a)}.$ That is, $$\beta_{\rhd_T,a}b=\alpha_{\rhd_T,S_T(a)}b=S_T(a)\rhd_Tb.$$
Then we have
\begin{eqnarray*}
\alpha_{\rhd_T,a_1}\beta_{\rhd_T,a_2}b & =&T(a_1)\rightharpoonup (T(S_T(a_2))\rightharpoonup b)\\
&=& T(a_1)T(S_T(a_2))\rightharpoonup b\\
&\stackrel{\eqref{eq:rrb-Hopf}}{=}& T(a_1(T(a_2)\rightharpoonup S_T(a_3)))\rightharpoonup b\\
&=& T(a_1(T(a_2)\rightharpoonup (S_H(T(a_3))\rightharpoonup S_K(a_4))))\rightharpoonup b\\
&=& T(a_1(T(a_2)S_H(T(a_3))\rightharpoonup S_K(a_4)))\rightharpoonup b\\
&=& T(a_1S_K(a_2))\rightharpoonup b\\
&=& T(\vep_K(a)1)\rightharpoonup b=\vep_K(a)b,\\[.5em]
\beta_{\rhd_T,a_1}\alpha_{\rhd_T,a_2}b & =&T(S_T(a_1))\rightharpoonup (T(a_2)\rightharpoonup b)\\
&=& T(S_T(a_1))T(a_2)\rightharpoonup b\\
&\stackrel{\eqref{eq:rrb-anti}}{=}& S_H(T(a_1))T(a_2)\rightharpoonup b\\
&=& \vep_H(T(a))b=\vep_K(a)b.
\end{eqnarray*}
Therefore, $\alpha_{\rhd_T}$ is convolution invertible. Hence, $(K,\rhd_T)$ is a post-Hopf algebra.

Let $(f,g)$ be a homomorphism from the relative Rota-Baxter operator $T$ to $T'$. Then we have
\begin{eqnarray*}
  g(a\rhd_T b)=g(T(a)\rightharpoonup b)=f(T(a))\rightharpoonup g(b)=T'(g(a))\rightharpoonup g(b)=g(a) \rhd _{T'} g(b),
\end{eqnarray*}
which implies that $g$ is a homomorphism from the  post-Hopf algebra  $(K,\rhd_T)$ to $(K',\rhd_{T'})$.  It is straightforward to see that this is indeed a functor.

Next we prove that $\Xi|_{\cocrRB}:\cocrRB\to\cocPH$ is right adjoint to $\Upsilon:\cocPH\to \cocrRB$.
Namely,
$$\Hom_\cocrRB(\id:H'\to H'_{\rhd'},T:K\to H)\simeq\Hom_\cocPH((H',\rhd'),(K,\rhd_T)),$$
where $T:K\to H$ is a relative Rota-Baxter operator on a Hopf algebra $H$ with respect to a cocommutative module bialgebra  $(K,\rightharpoonup)$ and $(H',\rhd')$ is a cocommutative post-Hopf algebra.

Let $g:(H',\rhd')\to (K,\rhd_T)$ be a post-Hopf algebra homomorphism. Let $f=Tg$, which  is obviously  a coalgebra homomorphism.  For all $x,y\in H'$, we have
\begin{align*}
f(x*_{\rhd'}y)&=T(g(x_1(x_2\rhd'y)))=T(g(x_1)(g(x_2)\rhd_T g(y)))\\
&=T(g(x_1)(T(g(x_2))\rightharpoonup g(y)))=T(g(x))T(g(y))=f(x)f(y),
\end{align*}
which implies that $f=T g:H'_{\rhd'}\to H$ is a Hopf algebra homomorphism.

 Moreover, it is straightforward to obtain
$$g(x\rhd' y)=g(x)\rhd_T g(y)=T(g(x))\rightharpoonup g(y)=f(x)\rightharpoonup g(y).$$
Hence, $(f,g)$ is a homomorphism between the relative Rota-Baxter operators $\id:H'\to H'_{\rhd'}$ and $T:K\to H$.

Conversely, if $(f,g)$ is a homomorphism between the relative Rota-Baxter operators $\id:H'\to H'_{\rhd'}$ and $T:K\to H$, we  have $f=Tg$ and $g:(H',\rhd')\to (K,\rhd_T)$ is a post-Hopf algebra homomorphism.
\end{proof}

By Theorem \ref{th:post-action} and Theorem \ref{thm:subHopf}, we immediately get the following result.

\begin{coro}\label{prop:rrb-des}
Let $T:K\to H$ be a relative Rota-Baxter operator
with respect to a cocommutative $H$-module bialgebra $(K,\rightharpoonup)$. Then $(K,*_T,1,\Delta,\vep,S_T)$ is a Hopf algebra, which is called the {\bf descendent Hopf algebra} and denoted by $K_T$, where  the antipode $S_T$ is given by \eqref{rrb-des2} and   the multiplication $*_T$ is given by
\begin{eqnarray}
\label{rrb-des1}a *_T b&=&a_1(T(a_2)\rightharpoonup b).
\end{eqnarray}
Moreover,  $T:K_T\to H$ is a Hopf algebra homomorphism.
\end{coro}

\section{Matched pairs of Hopf algebras and solutions of the Yang-Baxter equation}\label{sec:mp}

In this section, we show that a relative Rota-Baxter operator  on cocommutative Hopf algebras naturally gives rise to a matched pair of Hopf algebras. As applications, we construct solutions of the Yang-Baxter equation  using post-Hopf algebras and relative Rota-Baxter operators   on cocommutative Hopf algebras.

First we recall the smash product  and matched pairs of Hopf algebras.
Let $H$ and $K$ be two Hopf algebras such that $K$ is a cocommutative  $H$-module bialgebra via an action $\rightharpoonup$.
There is the following {\bf smash product} on $K\otimes H$,
$$(a\# x)(a'\# x')=a(x_1\rightharpoonup a')\# x_2x'$$
for any $x,x'\in H,\,a,a'\in K$, where $a\otimes x\in K\otimes H$ is rewritten as $a\# x$ to emphasize this smash product. We denote such a smash product algebra by $K\rtimes H$.
In particular, if $H$ is also cocommutative, then $K\rtimes H$ becomes a cocommutative Hopf algebra with the usual tensor product comultiplication and the antipode defined by $S(a\# x)=(S_H(x_1)\rightharpoonup S_K(a))\#S_H(x_2)$.
\begin{defn}
A {\bf matched pair of Hopf algebras} is a 4-tuple $(H,K,\rightharpoonup,\leftharpoonup)$, where $H$ and $K$
are Hopf algebras, $\rightharpoonup:H\otimes K\to K$ and $\leftharpoonup:H\otimes K\to H$ are linear maps such that $K$ is a left $H$-module coalgebra and $H$ is a right $K$-module coalgebra  and the following compatibility
conditions hold:
\begin{eqnarray}
\label{Mat-1}x\rightharpoonup(ab)&=&(x_1\rightharpoonup a_1)((x_2\leftharpoonup a_2)\rightharpoonup b)\\
\label{Mat-2}x\rightharpoonup 1_K&=&\varepsilon_H(x)1_K\\
\label{Mat-3}(xy)\leftharpoonup a&=&(x\leftharpoonup(y_1\rightharpoonup a_1))(y_2\leftharpoonup a_2)\\
\label{Mat-4}1_H\leftharpoonup a&=&\varepsilon_K(a)1_H\\
\label{Mat-5} (x_1\leftharpoonup a_1) \otimes (x_2\rightharpoonup a_2) &=&
(x_2\leftharpoonup a_2) \otimes (x_1\rightharpoonup a_1)
\end{eqnarray}
for all $x,y\in H$ and $a,b\in K$.
\end{defn}

Let $(H,K,\rightharpoonup,\leftharpoonup)$ be a matched pair of Hopf algebras. The  {\bf double crossproduct} $K\bowtie H$ of $K$ and $H$ is the $\bk$-vector space $K\otimes H$ with the unit
$1_K\otimes 1_H$, such that its product,  coproduct, counit and antipode are given by
\begin{eqnarray}
(a\otimes x)(b\otimes y)&=&a(x_1\rightharpoonup b_1)\otimes (x_2\leftharpoonup b_2)y,\\
\Delta(a\otimes x)&=&a_1\otimes x_1\otimes a_2\otimes x_2,\\
\varepsilon(a\otimes x)&=&\varepsilon_K(a)\varepsilon_H(x),\\
S(a\otimes x)&=&(S_H(x_2)\rightharpoonup S_K(a_2))\otimes(S_H(x_1)\leftharpoonup S_K(a_1)),
\end{eqnarray}
for all $a,b\in K$ and $x,y\in H$. See \cite{Maj1} for further details of the double crossproducts.

 By \cite[Proposition~21.6]{Maj1}, we have
\begin{theorem}\label{Prop:mp-Hopf}
With above notations, $(H,K,\rightharpoonup,\leftharpoonup)$ is a matched pair of Hopf algebras if and only if there exist a Hopf algebra $A$ and injective Hopf algebra homomorphisms $i_K:K\to A$, $i_H:H\to A$ such that
the map
$$\xi:K\otimes H\to A,\,a\otimes x\mapsto i_K(a)i_H(x)$$
is a linear isomorphism.
\end{theorem}

\emptycomment{
\begin{proof}
If $(H,K,\rightharpoonup,\leftharpoonup)$ is a matched pair of Hopf algebras, we can take $A=K\bowtie H$, $i_K,\,i_H$ the natural inclusions of $K,H$ in $A$
respectively. Then $\xi$ is the identity isomorphism of $K\otimes H$.

Conversely, note that $\xi$ is clearly a coalgebra isomorphism.
First define the {\it twistor}
$$\Theta:H\otimes K\to K\otimes H,\,x\otimes a\mapsto\xi^{-1}(i_H(x)i_K(a)),\quad \forall a\in K,\,x\in H.$$
Then $(\vep_K\otimes\vep_H)\Theta(x\otimes a)=\vep_H(x)\vep_K(a)$, and
$$
\Delta\Theta(x\otimes a)=\Delta\xi^{-1}(i_H(x)i_K(a))=\xi^{-1}(i_H(x_1)i_K(a_1))\otimes
\xi^{-1}(i_H(x_2)i_K(a_2))=(\Theta\otimes \Theta)\Delta(x\otimes a),$$
$\Theta$ is also a coalgebra homomorphism. Now we define linear maps
\begin{align*}
&\rightharpoonup:H\otimes K\to K,\,x\otimes a\mapsto (\id\otimes\vep_H)\Theta(x\otimes a),\\
&\leftharpoonup:H\otimes K\to H,\,x\otimes a\mapsto (\vep_K\otimes\id)\Theta(x\otimes a).
\end{align*}
In particular,
\begin{align*}
\Theta(x\otimes a)&=(\id\otimes\vep_H\otimes \vep_K\otimes\id)\Delta\Theta(x\otimes a)\\
&=(\id\otimes\vep_H)\Theta(x_1\otimes a_1)\otimes(\vep_K\otimes\id)\Theta(x_2\otimes a_2)\\
&=(x_1\rightharpoonup a_1) \otimes (x_2\leftharpoonup a_2),
\end{align*}
or equivalently, $i_H(x)i_K(a)=i_K(x_1\rightharpoonup a_1)i_H(x_2\leftharpoonup a_2)$.
Also, $\rightharpoonup,\leftharpoonup$ are coalgebra homomorphisms, as
\begin{align*}
\Delta_K(x\rightharpoonup a)&=(\Delta_K\otimes\vep_H)\Theta(x\otimes a)=
(\id\otimes\vep_K\otimes\id\otimes\vep_H)\Delta\Theta(x\otimes a)\\
&=(\id\otimes\vep_K)\Theta(x_1\otimes a_1)\otimes(\id\otimes\vep)\Theta(x_2\otimes a_2)=(x_1\rightharpoonup a_1)\otimes(x_2\rightharpoonup a_2),\\
\Delta_H(x\leftharpoonup a)&=(\vep_K\otimes\Delta_H)\Theta(x\otimes a)=
(\vep\otimes\id\otimes\vep\otimes\id)\Delta\Theta(x\otimes a)\\
&=(\vep_K\otimes\id)\Theta(x_1\otimes a_1)\otimes(\vep_K\otimes\id)\Theta(x_2\otimes a_2)=(x_1\leftharpoonup a_1)\otimes(x_2\leftharpoonup a_2),\\
\vep_K(x\rightharpoonup a)&=\vep_H(x\leftharpoonup a)=(\vep_K\otimes\vep_H)\Theta(x\otimes a)=\vep_H(x)\vep_K(a),
\end{align*}
and
\begin{align*}
&x\rightharpoonup 1_K=(\id\otimes\vep_H)\Theta(x\otimes 1_K)=\vep_H(x)1_K,\\
&1_H\leftharpoonup a=(\vep_K\otimes\id)\Theta(1_H\otimes a)=\vep_K(a)1_H,
\end{align*}
which are Eqs.~\eqref{Mat-2} and \eqref{Mat-4}. Furthermore, we have
\begin{align*}
((x_1y_1)\rightharpoonup a_1)\otimes((x_2y_2)\leftharpoonup a_2)
&=\Theta(xy\otimes a)\\
&=\xi^{-1}(i_H(xy)i_K(a))=\xi^{-1}(i_H(x)i_H(y)i_K(a))\\
&=\xi^{-1}(i_H(x)i_K(y_1\rightharpoonup a_1)i_H(y_2\leftharpoonup a_2))\\
&=\xi^{-1}(i_K(x_1\rightharpoonup(y_1\rightharpoonup a_1))
i_H(x_2\leftharpoonup(y_2\rightharpoonup a_2))i_H(y_3\leftharpoonup a_3))\\
&=\xi^{-1}(i_K(x_1\rightharpoonup(y_1\rightharpoonup a_1))
i_H((x_2\leftharpoonup(y_2\rightharpoonup a_2))(y_3\leftharpoonup a_3)))\\
&=(x_1\rightharpoonup(y_1\rightharpoonup a_1))\otimes (x_2\leftharpoonup(y_2\rightharpoonup a_2))(y_3\leftharpoonup a_3).
\end{align*}
Applying $\id\otimes\vep_H$ (resp. $\vep_K\otimes \id$) to both hand sides, we see that $\rightharpoonup$ is a left action (resp. Eq.~\eqref{Mat-3} holds), and thus $K$ is a left $H$-module coalgebra via $\rightharpoonup$.
By symmetry, we also know that $H$ is a right $K$-module coalgebra via $\leftharpoonup$ and Eq.~\eqref{Mat-1} holds.

On the other hand,
\begin{align*}
(x_1\leftharpoonup a_1) \otimes (x_2\rightharpoonup a_2)
&=
(\vep_K\otimes\id)\Theta(x_1\otimes a_1)\otimes(\id\otimes\vep_H)\Theta(x_2\otimes a_2)\\
&=(\vep_K\otimes\id\otimes\id\otimes\vep_H)\Delta\Theta(x\otimes a)\\
&=(\vep_K\otimes\id\otimes\id\otimes\vep_H)\Delta^{\rm op}\Theta(x\otimes a)\\
&=(\vep_K\otimes\id)\Theta(x_2\otimes a_2)\otimes(\id\otimes\vep_H)\Theta(x_1\otimes a_1)\\
&=(x_2\leftharpoonup a_2) \otimes (x_1\rightharpoonup a_1),
\end{align*}
which is Eq.~\eqref{Mat-5}.
\end{proof}
}

Let $T:K\to H$ be a relative Rota-Baxter operator
with respect to a cocommutative $H$-module bialgebra $(K,\rightharpoonup)$.
Define a linear map $\leftharpoonup:H\otimes K\to H$ by
\begin{eqnarray}
\label{eq:mp-raction}
 x\leftharpoonup a&=&S_H\big(T(x_1\rightharpoonup a_1)\big)x_2T(a_2).
\end{eqnarray}
\begin{theorem}\label{mat-rrb-twist}
With the above notations, if $H$ is also cocommutative, then it is a right $K_T$-module coalgebra via the action $\leftharpoonup$ given in Eq.~\eqref{eq:mp-raction}. Moreover, the $4$-tuple $(H,K_T,\rightharpoonup,\leftharpoonup)$ is a matched pair of cocommutative Hopf algebras.
\end{theorem}
\begin{proof}
We define a linear map $\Phi_T:K\otimes H\to K\otimes H$ as following:
\begin{eqnarray}
\label{twist-Hopf}\Phi_T(a\otimes x)=a_1\otimes T(a_2)x,~\forall x\in H,\,a\in K.
\end{eqnarray}
Since $T$ is a coalgebra homomorphism, the linear map $\Phi_T$ is invertible. Moreover, we have
\begin{eqnarray}
\Phi^{-1}_T(a\otimes x)=a_1\otimes S_H(T(a_2))x,~\forall x\in H,\,a\in K.
\end{eqnarray}
Transfer the smash product Hopf algebra structure $K\rtimes H$ to   $K\otimes H$ via the linear isomorphism  $\Phi_T:K\otimes H\to K\rtimes H$, we obtain a Hopf algebra   $(K\otimes H,\cdot_T,1_T,\Delta_T,\varepsilon_T,\mathfrak S_T)$. Denote elements in $K\otimes H$  by $a\bowtie x,~b\bowtie y$ for  $x,y\in H,\,a,b\in K$, by the cocommutativity of $K$, we have
\begin{eqnarray*}
(a\bowtie x)\cdot_T(b\bowtie y)&=&\Phi^{-1}_T\big(\Phi_T(a\bowtie x)\Phi_T(b\bowtie y)\big)\\
&=&\Phi_T^{-1}\big(a_1(T(a_2)x_1\rightharpoonup b_1)\# T(a_3)x_2T(b_2)y\big)\\
                                                   &=&a_1(T(a_2)x_1\rightharpoonup b_1)\bowtie S_H\big(T(a_3(T(a_4)x_2\rightharpoonup b_2))\big)T(a_5)x_3T(b_3)y\\
                                                   &=&a_1(T(a_2)\rightharpoonup(x_1\rightharpoonup b_1))\bowtie S_H\big(T(a_3(T(a_4)\rightharpoonup(x_2\rightharpoonup b_2)))\big)T(a_5)x_3T(b_3)y\\
                                                   &\stackrel{\eqref{eq:rrb-Hopf}}{=}&a_1*_T(x_1\rightharpoonup b_1)\bowtie S_H\big(T(a_2)T(x_2\rightharpoonup b_2)\big)T(a_3)x_3T(b_3)y\\
                                                   &=&a_1*_T(x_1\rightharpoonup b_1)\bowtie S_H\big(T(x_2\rightharpoonup b_2)\big)S_H(T(a_2))T(a_3)x_3T(b_3)y\\
                                                   &=&a*_T(x_1\rightharpoonup b_1)\bowtie S_H\big(T(x_2\rightharpoonup b_2)\big)x_3T(b_3)y\\
                                                   &=&a\ast_T(x_1\rightharpoonup b_1)\bowtie (x_2\leftharpoonup b_2)y,\\[.4em]
\Delta_T(a\bowtie x)&=&(\Phi^{-1}_T\otimes \Phi^{-1}_T)(\Delta\Phi_T(a\bowtie x))\\
                      &=&(\Phi^{-1}_T\otimes \Phi^{-1}_T)\Delta(a_1\# T(a_2)x)\\
                      &=&(\Phi^{-1}_T\otimes \Phi^{-1}_T)\big((a_1\#
                       T(a_3)x_1)\otimes (a_2\# T(a_4)x_2)\big)\\
                      &=&(a_1\bowtie S_H(T(a_2))T(a_5)x_1)\otimes (a_3\bowtie S_H(T(a_4))T(a_6)x_2)\\
                      &=&(a_1\bowtie S_H(T(a_2))T(a_3)x_1)\otimes (a_4\bowtie S_H(T(a_5))T(a_6)x_2)\\
                      &=&(\varepsilon_K(a_2)a_1\bowtie x_1)\otimes (\varepsilon_K(a_4)a_3\bowtie x_2)\\
                      &=&(a_1\bowtie x_1)\otimes (a_2\bowtie x_2),\\[.4em]
\mathfrak S_T(a\bowtie x) &=& \Phi^{-1}_T(S_\rtimes\Phi_T(a\bowtie x))\\
 &=& \Phi^{-1}_T(S_\rtimes(a_1\# T(a_2)x))\\
 &=& \Phi^{-1}_T((S_H(T(a_1)x_1)\rightharpoonup S_K(a_2))\#S_H(T(a_3)x_2))\\
 &=& \Phi^{-1}_T((S_H(x_1)\rightharpoonup(S_H(T(a_1))\rightharpoonup S_K(a_2)))\#S_H(T(a_3)x_2))\\
 &\stackrel{\eqref{rrb-des2}}{=}& \Phi^{-1}_T((S_H(x_1)\rightharpoonup S_T(a_1))\#S_H(x_2)S_H(T(a_2))\\
 &=& (S_H(x_1)\rightharpoonup S_T(a_1))\bowtie S_H(T(S_H(x_2)\rightharpoonup S_T(a_2)))S_H(x_3)S_H(T(a_3))\\
 &=& (S_H(x_1)\rightharpoonup S_T(a_1))\bowtie S_H(T(S_H(x_2)\rightharpoonup S_T(a_2)))S_H(x_3)T(S_T(a_3))\\
 &\stackrel{\eqref{eq:mp-raction}}{=}&  (S_H(x_1)\rightharpoonup S_T(a_1))\bowtie (S_H(x_2)\leftharpoonup S_T(a_2)).
\end{eqnarray*}
Moreover, it is obvious that $1_T=1\bowtie 1$  and $$\varepsilon_T(a\bowtie x)=\varepsilon_K(a_1)\varepsilon_H(T(a_2)x)=\varepsilon_K(a)\varepsilon_H(x).$$
Define linear maps $\frki_{K_T}:K_T\to K\otimes H$ and $\frki_{H}:H\to K\otimes H$ by
$$
\frki_{K_T}(a)=a\bowtie 1,\quad
\frki_{H}(x)=1\bowtie x.
$$
Then it is obvious  that $\frki_{K_T}$ and $\frki_{H}$ are injective Hopf algebra homomorphisms, and
\begin{eqnarray*}
(a\bowtie 1)\cdot_{T}(1\bowtie x)&=&a\ast_T(1\rightharpoonup 1)\bowtie (1\leftharpoonup 1)x=a_1(T(a_2)\rightharpoonup 1)\bowtie x\\
&=&a_1\vep_H(T(a_2))\bowtie x=a\bowtie x.
\end{eqnarray*}
Therefore, we obtain that $(K\otimes H,\cdot_T,1_T,\Delta_T,\varepsilon_T,\mathfrak S_T)$ is a Hopf algebra that can be factorized into Hopf algebras $K_T$ and $H$. Thus, we deduce that $H$ is a right $K_T$-module coalgebra via the action $\leftharpoonup$ and $K_T$ is a left $H$-module coalgebra via the action $\rightharpoonup $ and the 4-tuples $(H,K_T,\rightharpoonup,\leftharpoonup)$ is a matched pair of Hopf algebras by Theorem~\mref{Prop:mp-Hopf}. Moreover, the Hopf algebra $(K\otimes H,\cdot_T,1_T,\Delta_T,\varepsilon_T,\mathfrak S_T)$ is exactly the double crossproduct $K_T\bowtie H.$
\end{proof}

 Conversely, let $H$ and $K$ be two cocommutative Hopf algebras such that $K$ is an $H$-module bialgebra via an action $\rightharpoonup$. Let $T:K\to H$ be a coalgebra homomorphism, and $(K\otimes H,\cdot_T,1_T,\Delta_T,\varepsilon_T,\mathfrak S_T)$   the Hopf algebra   obtained from the smash product $K\rtimes H$ via the linear isomorphism $\Phi_T$ given in \eqref{twist-Hopf}.
\begin{prop}
If $K\otimes 1$ is a subalgebra of the Hopf algebra $(K\otimes H,\cdot_T,1_T,\Delta_T,\varepsilon_T,\mathfrak S_T)$,
then $T$ is a relative Rota-Baxter operator
with respect to the $H$-module bialgebra $(K,\rightharpoonup)$.
\end{prop}
\begin{proof}
Since $K\otimes 1$ is a subalgebra of $(K\otimes H,\cdot_T,1_T,\Delta_T,\varepsilon_T,\mathfrak S_T)$, for any $a,b\in K$ we have
\begin{eqnarray*}
(a\bowtie 1)\cdot_T(b\bowtie 1)&=&\Phi^{-1}_T\big(\Phi_T(a\bowtie 1)\Phi_T(b\bowtie 1)\big)\\
&=&\Phi^{-1}_T\big(a_1(T(a_2)\rightharpoonup b_1)\# T(a_3)T(b_2)\big)\\
                                                   &=&a_1(T(a_2)\rightharpoonup b_1)\bowtie S_H\big(T(a_3(T(a_4)\rightharpoonup b_2))\big)T(a_5)T(b_3)\in K\otimes 1.
\end{eqnarray*}
Applying $m_H(T\otimes\id)$ and $T\otimes\vep_H$ to it respectively, we obtain that
\begin{eqnarray*}
T(a)T(b)&=&T(a_1(T(a_2)\rightharpoonup b_1))S_H\big(T(a_3(T(a_4)\rightharpoonup b_2))\big)T(a_5)T(b_3)\\
&=&T(a_1(T(a_2)\rightharpoonup b_1))\vep_H(S_H\big(T(a_3(T(a_4)\rightharpoonup b_2))\big)T(a_5)T(b_3))\\
&=&T(a_1(T(a_2)\rightharpoonup b_1)).
\end{eqnarray*}
Namely,  \eqref{eq:rrb-Hopf} holds, and $T$ is a relative Rota-Baxter operator.
\end{proof}


Let  $(H,\rhd)$ be a post-Hopf algebra and $H_\rhd\coloneqq(H,*_\rhd,1,\Delta,\vep,S_\rhd)$ the subadjacent Hopf algebra given in Theorem \ref{thm:subHopf}. By Proposition \ref{prop:PHtoRB}, the identity map $\id:H\to H_\rhd$ is a relative Rota-Baxter operator. By Theorem   \ref{mat-rrb-twist}, we have

\begin{coro}\label{cor:post-mp}
Let  $(H,\rhd)$ be a cocommutative post-Hopf algebra. Then the $4$-tuple $(H_\rhd,H_\rhd,\rhd,\lhd)$ is a matched pair of cocommutative Hopf algebras, where    $\lhd$ is given by
\begin{eqnarray}
a\lhd b&=&S_\rhd\big(a_1\rhd b_1\big)\ast_\rhd a_2 \ast_\rhd b_2.
\end{eqnarray}
Moreover, we have  the compatibility condition
\begin{eqnarray}\label{braiding-condition}
a\ast_\rhd b=(a_1\rhd b_1)\ast_\rhd(a_2 \lhd b_2).
\end{eqnarray}
\end{coro}

\begin{proof}
  We only  need  to check the stated compatibility condition, which follows from
\begin{align*}
(a_1\rhd b_1)\ast_\rhd(a_2 \lhd b_2)
&=(a_1\rhd b_1)\ast_\rhd\big(S_\rhd\big(a_2\rhd b_2\big)\ast_\rhd a_3 \ast_\rhd b_4\big)\\
&=\big((a_1\rhd b_1)\ast_\rhd S_\rhd\big(a_2\rhd b_2\big)\big)\ast_\rhd a_3 \ast_\rhd b_4\\
&=\vep(a_1\rhd b_1)a_2 \ast_\rhd b_2\\
&=\vep(a_1)\vep(b_1)
a_2 \ast_\rhd b_2\\
&=a \ast_\rhd b.\qedhere
\end{align*}
\end{proof}

At the end of this section, we show that post-Hopf algebras and relative Rota-Baxter operators   on cocommutative Hopf algebras give rise to solutions of the Yang-Baxter equation.

\begin{defn}
A  solution of the Yang-Baxter equation  on a vector space $V$
is an invertible linear endomorphism $R : V \otimes  V\to  V \otimes V$ such that
$$(R \otimes \id)(\id \otimes R)(R \otimes \id) = (\id \otimes R)(R \otimes \id)(\id \otimes R).$$
\end{defn}

\emptycomment{
A set-theoretical solution to the Yang-Baxter equation (YBE) on a set X
is a bijective map $S : X \times  X\to  X \times X$ such that
$$(S \times \id)(\id \times S)(S \times \id) = (\id \times S)(S \times \id)(\id \times S).$$
}

\begin{theorem}\label{thm:post-YBE}
Let  $(H,\rhd)$ be a cocommutative post-Hopf algebra. Then $R:H\otimes H\to H\otimes H$ defined by
$$
R(x\otimes y)= (x_1\rhd y_1 )\otimes  (x_2\lhd y_2 ),
$$
where $\lhd$ is defined by \eqref{eq:mp-raction}, is a coalgebra isomorphism and a solution of the Yang-Baxter equation on the vector space $H$.
\end{theorem}

\begin{proof}
Denote by $H_\rhd^l$ and $H_\rhd^r$ two copies of the Hopf algebra $H_\rhd$. By Corollary \ref{cor:post-mp},   $(H_\rhd^l,H_\rhd^r,\rhd,\lhd)$ is a matched pair of cocommutative Hopf algebras.  Thus, $A=H_\rhd^l\bowtie H_\rhd^r$ is a Hopf  algebra that factorized into Hopf algebras $H_\rhd^l$ and $H_\rhd^r$. By Theorem \ref{Prop:mp-Hopf}, there is a coalgebra isomorphism $$\xi:H_\rhd^l\otimes H_\rhd^r\stackrel{\frki_{H_\rhd^l}\otimes \frki_{H_\rhd^r}}{\longrightarrow}A\otimes A\stackrel{m_{H_\rhd^l\bowtie H_\rhd^r}}{\longrightarrow} A=H_\rhd^l\bowtie H_\rhd^r.$$ We consider the coalgebra homomorphism 
 $$\Psi=\xi^{-1}\circ m_{H_\rhd^l\bowtie H_\rhd^r}\circ (\frki_{H_\rhd^r}\otimes \frki_{H_\rhd^l}):H_\rhd^r\otimes H_\rhd^l\to H_\rhd^l\otimes H_\rhd^r.$$
Since $H_\rhd^l\bowtie H_\rhd^r$ is a cocommutative Hopf algebra, we deduce that $\Psi$ is a coalgebra isomorphism. Moreover, $\Psi$ satisfies the following equations:
\begin{eqnarray*}
\Psi\circ(m_{H_\rhd}\otimes \Id)&=&(\Id\otimes m_{H_\rhd})\circ (\Psi\otimes \Id)\circ(\Id\otimes \Psi),\\
\Psi\circ(\Id\otimes m_{H_\rhd})&=&(m_{H_\rhd}\otimes\Id)\circ (\Id\otimes \Psi)\circ(\Psi\otimes \Id),\\
\Psi(1\otimes x)&=&x\otimes 1,\\
\Psi(x\otimes 1)&=&1\otimes x.
\end{eqnarray*}
For all $x,y\in H$, we have
\begin{eqnarray*}
\Psi(x\otimes y)=(x_1\rhd y_1)\otimes(x_2\lhd y_2)=R(x\otimes y).
\end{eqnarray*}
By \eqref{braiding-condition}, we have $m_{H_\rhd}=m_{H_\rhd}\circ \Psi$. Thus, we deduce that $R=\Psi$ is a braiding operator on the cocommutative Hopf algebra $H_\rhd\coloneqq(H,*_\rhd,1,\Delta,\vep,S_\rhd)$. By \cite[Theorem 4.11]{GGV}, we obtain that $R$ is a solution of the Yang-Baxter equation on the vector space $H$.
\end{proof}

\begin{exam}
Consider the post-Hopf algebra $(T\bk\{\huaO\},\Delta^{\co},\rhd)$ given in Example \ref{free-post-lie}. Then $R:T\bk\{\huaO\}\otimes T\bk\{\huaO\}\to T\bk\{\huaO\}\otimes T\bk\{\huaO\}$ defined by
$$
R(\huaX\otimes \huaY)=\Big(\huaX_1\rhd \huaY_1\Big)\otimes \Big(\huaX_2\lhd \huaY_2\Big),\quad \huaX,\huaY\in T\bk\{\huaO\}
$$
is a coalgebra isomorphism and a solution of the Yang-Baxter equation on the vector space $T\bk\{\huaO\}$. More precisely, we have
 \begin{eqnarray*}
R(\huaX\otimes \huaY)=\Big(\huaX_1\rhd\huaY_1\Big)\otimes B^-\Big(S_{\rhd}\big(\huaX_2\rhd\huaY_2\big)\rhd\big(\huaX_3\rhd B^+(\huaY_3)\big)\Big),
\end{eqnarray*}
where ${\Delta^{\co}}^{(2)}\huaX=\huaX_1\otimes\huaX_2\otimes \huaX_3$ and ${\Delta^{\co}}^{(2)}\huaY=\huaY_1\otimes\huaY_2\otimes \huaY_3$.
\end{exam}

\begin{exam}
Consider the pre-Hopf algebra  $(S\bk\{\huaT\},\Delta^{\co},\rhd)$ given in Example \ref{free-pre-lie}. Then $R:S\bk\{\huaT\}\otimes S\bk\{\huaT\}\to S\bk\{\huaT\}\otimes S\bk\{\huaT\}$ defined by
$$
R(\huaX\otimes \huaY)=\Big(\huaX_1\rhd \huaY_1\Big)\otimes \Big(\huaX_2\lhd \huaY_2\Big),\quad \huaX,\huaY\in S\bk\{\huaT\}
$$
is a coalgebra isomorphism and a solution of the Yang-Baxter equation on the vector space $S\bk\{\huaT\}$. More precisely, for forests $\huaX,\huaY\in S\bk\{\huaT\}$, we have
\begin{eqnarray*}
R(\huaX\otimes \huaY)=\Big(\huaX_1\rhd\huaY_1\Big)\otimes B^-\Big(S_{\rhd}\big(\huaX_2\rhd\huaY_2\big)\rhd\big(\huaX_3\rhd B^+(\huaY_3)\big)\Big),
\end{eqnarray*}
where ${\Delta^{\co}}^{(2)}\huaX=\huaX_1\otimes\huaX_2\otimes \huaX_3$ and ${\Delta^{\co}}^{(2)}\huaY=\huaY_1\otimes\huaY_2\otimes \huaY_3$.

\end{exam}

Let $T:K\to H$ be a relative Rota-Baxter operator on   $H$ with respect to a commutative $H$-module bialgebra $(K,\rightharpoonup)$. By Theorem \ref{th:post-action}, $(K,\rhd_T)$ is a commutative post-Hopf algebra.  By Corollary~\ref{prop:rrb-des}, there is a descendent Hopf algebra $K_T=(K,*_T,\Delta,\vep,S_T)$, such that $K$ is a $K_T$-module bialgebra via the action $\rhd_T$ defined in  \eqref{eq:post-action}. By Corollary \ref{cor:post-mp}, we have

\begin{coro}
The $4$-tuples $(K_T,K_T,\rhd_T,\lhd_T)$ is a matched pair of cocommutative Hopf algebras, here  $\rhd_T$ is  given by \eqref{eq:post-action} and $\lhd_T$ is given by
\begin{eqnarray}
\label{RRB-2}a\lhd_T b&=&S_T\big(a_1\rhd_T b_1\big)\ast_T a_2 \ast_T b_2.
\end{eqnarray}
Moreover, we have  the compatibility condition
\begin{eqnarray}
a\ast_T b=(a_1\rhd_T b_1)\ast_T(a_2 \lhd_T b_2).
\end{eqnarray}
\end{coro}
\emptycomment{\begin{proof}
It still needs to check the stated compatibility condition, which follows from
\begin{align*}
(a_1\rhd_T b_1)\ast_T(a_2 \lhd_T b_2)
&=(a_1\rhd_T b_1)\ast_T\big(S_T\big(a_2\rhd_T b_2\big)\ast_T a_3 \ast_T b_4\big)\\
&=\big((a_1\rhd_T b_1)\ast_T S_T\big(a_2\rhd_T b_2\big)\big)\ast_T a_3 \ast_T b_4\\
&=\vep_K(a_1\rhd_T b_1)a_2 \ast_T b_2\\
&=\vep_K(a_1)\vep_K(b_1)
a_2 \ast_T b_2\\
&=a \ast_T b.\qedhere
\end{align*}
\end{proof}}

By Theorem \ref{thm:post-YBE}, we have

\begin{coro}
Let $T:K\to H$ be a relative Rota-Baxter operator  with respect to a commutative $H$-module bialgebra  $(K,\rightharpoonup)$. Then $R:K\otimes K\to K\otimes K$ defined by
$$
R(a\otimes b)=\Big(a_1\rhd_T b_1\Big)\otimes \Big(a_2\lhd_T b_2\Big)
$$
is a coalgebra isomorphism and a solution of the Yang-Baxter equation on the vector space $K$, where $\rhd_T$ and $\lhd_T$ are defined by \eqref{eq:post-action} and \eqref{RRB-2} respectively.
\end{coro}

\section{Equivalent characterizations of relative Rota-Baxter operators}\label{sec:gr}
In this section, we give some alternative characterizations of relative Rota-Baxter operators using relative Rota-Baxter operators on the Lie algebra of primitive elements, graphs and  module bialgebra structures.

\subsection{Restrictions and extensions of relative Rota-Baxter operators}Let $K$ be a  cocommutative $H$-module bialgebra via an action $\rightharpoonup$.  It is obvious that via the restrictions of the action $\rightharpoonup$, we obtain   actions of $G(H)$ on $G(K)$ and of $P(H)$ on $P(K)$, for which we use the same notations.
As expected, a  relative Rota-Baxter operator with respect to a cocommutative $H$-module bialgebra $(K,\rightharpoonup)$  will naturally induces a  relative Rota-Baxter operator on the group $G(H)$ and on the Lie algebra $P(H)$ respectively.

\begin{theorem}\label{thm:restriction}
Let $T:K\to H$ be a relative Rota-Baxter operator with respect to a cocommutative $H$-module bialgebra $(K,\rightharpoonup)$.
\begin{enumerate}[{\rm(i)}]
\item $T|_{G(K)}$ is a relative Rota-Baxter operator on the group $G(H)$ with respect to the action $(G(K),\rightharpoonup)$;
\item $T|_{P(K)}$ is a relative Rota-Baxter operator   on the Lie algebra $P(H)$ with respect to the action $(P(K),\rightharpoonup)$.
\end{enumerate}
\end{theorem}
\begin{proof}
Since $T$ is a coalgebra homomorphism, it follows  that $T|_{G(K)}$ is a map from $G(K)$ to $G(H)$, and $T|_{P(K)}$ is a map from $P(K)$ to $P(H)$.

 For any $a,b\in G(K)$, we have
$$T(a)T(b)=T\big(a(T(a)\rightharpoonup b)\big),$$
 which implies that $T|_{G(K)}$ is a  relative Rota-Baxter operator on the group $G(H)$ with respect to the action $(G(K),\rightharpoonup)$.

For any $a,b\in P(K)$, we have
$T(a)T(b)=T(ab)+T\big(T(a)\rightharpoonup b\big),$
and thus
$$[T(a),T(b)]=T\big(T(a)\rightharpoonup b\big)
-T\big(T(b)\rightharpoonup a\big)+T([a,b]).$$
 Hence,   $T|_{P(K)}$ is a relative Rota-Baxter operator on  the Lie algebra $P(H)$ with respect to the action $(P(K),\rightharpoonup)$.
\end{proof}

Let $\phi:\g\to\Der(\h)$ be an action of a Lie algebra $(\g,[\cdot,\cdot]_\g)$ on  $(\mathcal{\h},[\cdot,\cdot]_\h)$. Then $\phi$ can be extended to a module bialgebra action $\bar\phi:U(\g)\to \End(T_\bk(\h))$ by 
$$
\bar\phi(x)(1)=0,\,\bar\phi(x)(y_1\cdots y_r)=\sum_{i=1}^r y_1\cdots y_{i-1} \phi(x)(y_i) y_{i+1} \cdots y_r,
 $$
 where $T_\bk(\h)$ is the tensor $\bk$-algebra of $\h$, $x\in\g$ and $y_1,\dots, y_r\in\h,\,r\geq1$. As $\g$ acts on $\h$ by derivations, it induces a module bialgebra action $\bar\phi$ of $U(\g)$ on $U(\h)$.

The following extension theorem of relative Rota-Baxter operators from Lie algebras to their universal enveloping algebras generalizes the case of Rota-Baxter operators given in \cite[Theorem~2]{Go}.

\begin{theorem}\mlabel{prop:rrb-uea}
Any relative Rota-Baxter operator $T:\h\to \g$ on a Lie algebra $\g$ with respect to an action $(\h;\phi)$ can be extended to a unique relative Rota-Baxter operator $\bar T:U(\h)\to U(\g)$ with respect to the extended $U(\g)$-module bialgebra  $ ( U(\h),\bar\phi)$  by
$$
\bar T(y_1\cdots y_n)=\big(T(y_1)\bar T-\bar T\bar\phi(T(y_1))\big)\cdots\big(T(y_n)\bar T-\bar T\bar\phi(T(y_n))\big)(1),
\quad\forall y_1,\dots,y_n\in\h,\,n\geq1,$$
where those $T(y_k)$'s left to $\bar T$   are interpreted as the left multiplication by them.

Furthermore, the post-Hopf algebra $(U(\h),\rhd_{\bar T})$ induced by the relative Rota-Baxter operator $\bar T:U(\h)\to U(\g)$ as in Theorem~\ref{th:post-action} coincides with the extended post-Hopf algebra $(U(\h),\bar\rhd_T)$ from $(\h,\rhd_T)$ given in Theorem~\ref{th:post-uea}. Namely, we have the following diagram
\begin{equation}
\small{
 \xymatrix{(\h,\rhd_T) \ar[rr]^{\text{extension}} &  &  (U(\h),\rhd_{\bar T})    \\
\h\stackrel{T}{\to} \g \ar[rr]^{\text{extension}}\ar[u]^{\text{Rota-Baxter action}} & &U(\h)\stackrel{\bar{T}}{\to}U(\g)\ar[u]_{\text{Rota-Baxter action}}.}
}
\end{equation}
\end{theorem}
\begin{proof}
Let $J_\h=(yz-zy-[y,z]_\h\,|\,y,z\in\h)$ be the  ideal of $T_\bk(\h)$ such that $U(\h)\simeq T_\bk(\h)/J_\h$.
We recursively define a linear map $\bar T:T_\bk(\h)\to U(\g)$ by
\begin{equation}\label{eq:rec}
\bar T(1)=1,\quad \bar T(yu)=T(y)\bar T(u)-\bar T\big(\bar\phi(T(y))u\big),\quad\forall y\in \h,\,u\in \h^{\otimes n},\,n\geq0.
\end{equation}
Then it is straightforward to deduce that $\bar T(J_\h)=0$ and we have the induced linear map $\bar T:U(\h)\to U(\g)$.

Next we  prove that $\bar T:U(\h)\to U(\g)$ is a relative Rota-Baxter operator. Namely,
$$\bar T(u)\bar T(v)=\bar T\big(u_1\bar\phi(\bar T(u_2))v\big)$$
for any $u\in U(\h)_m,v\in U(\h)_n$. It can be done by induction on $m$. The case when $m=1$ is due to the recursive definition \eqref{eq:rec} of $\bar T$. For $yu\in U(\h)_{m+1}$, since $\bar\phi$ is a module bialgebra action, we have
\begin{align*}
\bar T(yu)\bar T(v)&=T(y)\bar T(u)\bar T(v)-\bar T\big(\bar\phi(T(y))u\big)\bar T(v)\\
&=T(y)\bar T\big(u_1\bar\phi(\bar T(u_2))v\big)
-\bar T\big(\big(\bar\phi(T(y))u_1\big)\big(\bar\phi(\bar T(u_2))v\big)\big)
-\bar T\big(u_1\big(\bar\phi\big(\bar T(\bar\phi(T(y))u_2)\big)v\big)\big)\\
&=\bar T\big(yu_1\bar\phi(\bar T(u_2))v\big)+\bar T\big(\bar\phi(T(y))\big(u_1\bar\phi(\bar T(u_2))v\big)\big)\\
&\quad-\bar T\big(\big(\bar\phi(T(y))u_1\big)\big(\bar\phi(\bar T(u_2))v\big)\big)
-\bar T\big(u_1\big(\bar\phi\big(\bar T(\bar\phi(T(y))u_2)\big)v\big)\big)\\
&=\bar T\big(yu_1\bar\phi(\bar T(u_2))v\big)+
\bar T\big(u_1\bar\phi\big(T(y)\bar T(u_2)\big)v\big)
-\bar T\big(u_1\big(\bar\phi\big(\bar T(\bar\phi(T(y))u_2)\big)v\big)\big)\\
&=\bar T\big(yu_1\bar\phi(\bar T(u_2))v\big)+
\bar T\big(u_1\bar\phi\big(\bar T(yu_2)\big)v\big)\\
&
=\bar T\big((yu)_1\bar\phi(\bar T((yu)_2))v\big),
\end{align*}
which implies that $\bar T:U(\h)\to U(\g)$ is a relative Rota-Baxter operator.
The above procedure also  implies that the extension   from $T:\h\to\g$ to $\bar T:U(\h)\to U(\g)$ is unique.

By  \eqref{eq:ind-post-L}, the induced post-Lie product $\rhd_T$ on $\h$ is given by
$$y \rhd_T z=\phi(T(y))z,\quad \forall y,z\in \h.$$
Then by Theorem~\ref{th:post-uea}, the extended post-Hopf product $\bar\rhd_T$
on $U(\h)$ is recursively defined by
\begin{align*}
&y\bar\rhd_T 1=0,\quad y\bar\rhd_T zv = (y\rhd_T z)v+z(y\bar\rhd_T v),\\
&1\bar\rhd_T v=v,\quad yu\bar\rhd_T v = y\bar\rhd_T(u \bar\rhd_T v)-(y\bar\rhd_T u)\bar\rhd_T v,
\end{align*}
for any $y,z\in\h,\,u,v\in U(\h)$.
On the other hand, by \eqref{eq:post-action}, we know that
$$u\rhd_{\bar T}v =\bar\phi(\bar T(u))v,\quad \forall u,v\in U(\h).$$
In particular, $y\rhd_{\bar T} 1=0$, $1\rhd_{\bar T}v=v$ and
\begin{align*}
y\rhd_{\bar T}zv&=\bar\phi(T(y))(zv)=(\phi(T(y))z)v+z\bar\phi(T(y))v\\
&=(y\rhd_T z)v +
z(y\rhd_{\bar T}v),\\
yu\rhd_{\bar T} v &= \bar\phi(\bar T(yu))v=\bar\phi\big(T(y)\bar T(u)-\bar T\big(\bar\phi(T(y))u\big)\big)v\\
&=\bar\phi(T(y))\big(\bar\phi(\bar T(u))v\big)-\bar\phi\big(\bar T\big(\bar\phi(T(y))u\big)\big)v\\
&=y\rhd_{\bar T} (u\rhd_{\bar T} v)-(y\rhd_{\bar T} u)\rhd_{\bar T} v.
\end{align*}
Therefore, the two post-Hopf products on $U(\h)$ coincide, and we get the desired diagram.
\end{proof}

\subsection{Graph characterization}

Now we use  graphs to characterize relative Rota-Baxter operators
  on  Hopf algebras.
\begin{defn}
Given any coalgebra homomorphism $f:K\to H$, we define the {\bf graph} of $f$, which is denoted by $\Gr_f$, as the subspace $\im((\id\otimes f)\Delta_K)$ of $K\otimes H$, namely,
$$\Gr_f=\{a_1\otimes f(a_2)\,|\,a\in K\}.$$
\end{defn}

\begin{theorem}\mlabel{th:rrb-gr}
A coalgebra homomorphism $T:K\to H$ is a relative Rota-Baxter operator  with respect to a cocommutative $H$-module bialgebra $(K,\rightharpoonup)$ if and only if the graph $\Gr_T$ is a Hopf subalgebra of the smash product Hopf algebra $K\rtimes H$ and isomorphic to $K_T$.
\end{theorem}
\begin{proof}

Let $T:K\to H$ be a relative Rota-Baxter operator. Then for all $a,b\in K$, we have
\begin{eqnarray*}
  (a_1\# T(a_2))(b_1\# T(b_2)) &=& a_1(T(a_2)\rightharpoonup b_1)\# T(a_3)T(b_2)\\
  &=& a_1\ast_T b_1\# T(a_2)T(b_2)\\
  &=& a_1\ast_T b_1\# T(a_2\ast_T b_2)\\
  &=&(a\ast_T b)_1\# T((a\ast_T b)_2)\in\Gr_T,
\end{eqnarray*}
as the binary operation $\ast_T$ on $K$ defined in \eqref{rrb-des1} is a coalgebra homomorphism by the cocommutativity of $K$, which implies that $\Gr_T$ is a subalgebra of $K\rtimes H$ with unit $1\#1=1\#T(1)$.

Also, as $T$ is a coalgebra homomorphism and $K$ is cocommutative,
\begin{eqnarray*}
\Delta(a_1\# T(a_2)) &=& (a_1\# T(a_3))\otimes (a_2\# T(a_4))\\
  &=& (a_1\# T(a_2))\otimes (a_3\# T(a_4))\in\Gr_T\otimes \Gr_T,\\
S(a_1\# T(a_2)) &=& (S_H(T(a_1))\rightharpoonup S_K(a_2))\#S_H(T(a_3))\\
 &\stackrel{\eqref{rrb-des2}}{=}& S_T(a_1)\# T(S_T(a_2))\\
  &=& S_T(a)_1\# T(S_T(a)_2)\in\Gr_T,
\end{eqnarray*}
where the antipode formula above is due to the property of $S_T$ shown in Proposition~\ref{prop:rrb-des}. Hence, $\Gr_T$ is a Hopf algebra inside $K\rtimes H$. Furthermore, define linear map
\begin{equation}\mlabel{eq:isom-rbb-gr}
\Psi:K_T\to \Gr_T,\,a\mapsto a_1\# T(a_2).
\end{equation}
By the calculation above, we know that $\Psi$ is a Hopf algebra isomorphism with its inverse $\Psi^{-1}=\id\otimes\vep_H$.

Conversely, assume that $\Gr_T$ is a subalgebra of $K\rtimes H$. For any $a,b\in K$, there exists $w\in K$ such that
$$a_1\ast_T b_1\# T(a_2)T(b_2) = (a_1\# T(a_2))(b_1\# T(b_2)) = w_1\# T(w_2).$$
In particular, since $\ast_T$ is bilinear, we have
\begin{align*}
w&=w_1\vep_H(T(w_2))=a_1\ast_T b_1\vep_H(T(a_2)T(b_2))=a\ast_T b,\\
T(w)&=\vep_K(w_1)T(w_2)= \vep_K(a_1\ast_T b_1)T(a_2)T(b_2) = T(a)T(b),
\end{align*}
which indicate 
that $T$ is a relative Rota-Baxter operator.
\end{proof}

 Let $\phi:\g\to\Der(\h)$ be an action of a Lie algebra  $(\g,[\cdot,\cdot]_\g)$  on   a Lie algebra   $(\mathcal{\h},[\cdot,\cdot]_\h)$.
 Denote by $\h\rtimes\g$ the semi-direct product Lie algebra  of $\h$ and $\g$ with respect to the   action $(\h,\phi)$. More precisely, the Lie bracket $[\cdot,\cdot]_{\rtimes}:\wedge^2(\h\oplus\g)\to \h\oplus\g$ is given by
\begin{eqnarray*}
     ~[(u,x),(v,y)]_{\rtimes}&=&([u,v]_\h+\phi(x)(v)-\phi(y)(u),[x,y]_\g),\quad \forall x,y\in \g, ~u,v\in \h.
\end{eqnarray*}
It is well known  that $T:\h\to\g$ is a relative Rota-Baxter operator if and only if the graph of $T$, $\Gr_T:=\{(u,T(u))\,|\,u\in\h\}$ is a subalgebra of $\h\rtimes\g$.

Now we consider the lifted relative Rota-Baxter operator  $\bar{T}:U(\h)\to U(\g)$  of the relative Rota-Baxter operator  $T:\h\to\g$. It turns out that the Hopf algebra $\Gr_{\bar{T}}$ can serve as the universal enveloping algebra of the Lie algebra $\Gr_T.$

\begin{prop}\mlabel{prop:Liegr-leftgr}
Let $\bar{T}:U(\h)\to U(\g)$ be the lifted relative Rota-Baxter operator  of the relative Rota-Baxter operator  $T:\h\to\g$. Then  $\Gr_{\bar{T}}\simeq U(\Gr_T),$ i.e. the universal enveloping algebra of the graph $\Gr_T$ of the relative Rota-Baxter operator $T:\h\to\g$    is isomorphic to the Hopf algebra $\Gr_{\bar{T}}$, which is the graph of the relative Rota-Baxter operator $\bar{T}:U(\h)\to U(\g)$.
\end{prop}
\begin{proof}
First it is straightforward to obtain the following Hopf algebra isomorphism
$$\bar\varphi:U(\h\rtimes\g)\to U(\h)\rtimes U(\g),\quad (u,x)\mapsto u\#1+1\#x,\quad\forall x\in \g, ~u\in \h.$$
By Theorem~\mref{th:rrb-gr}, $\Gr_{\bar T}$ is a Hopf subalgebra of $U(\h)\#U(\g)$ isomorphic to $U(\h)_{\bar T}$. Let $\psi\coloneqq \bar\varphi|_{\Gr_T}$, then $\psi$ is injective and
$$\psi(u,T(u))=u\#1+1\#T(u)
=(\id\otimes\bar T)\Delta(u)
\in \Gr_{\bar T},\quad \forall u\in\h.$$
Also, note that ${\rm Im}\,\psi$ generates $\Gr_{\bar T}$ as an algebra. Hence, $\psi$ induces a Hopf algebra isomorphism $\bar\psi:U(\Gr_T)\to \Gr_{\bar T}$ by the Theorem of Heyneman and Radford~\cite[Theorem 5.3.1]{Mon}, as $U(\Gr_T)_1=\bk\oplus \Gr_T$ and $\bar\psi|_{U(\Gr_T)_1}$ is also injective.
\end{proof}

\subsection{Module and module bialgebra characterization}

Next we give another characterization  of relative Rota-Baxter operators on Hopf algebras using new module structures and new module bialgebra structures. Let $H$ and $K$ be Hopf algebras such that $K$ is a cocommutative  $H$-module bialgebra via an action $\rightharpoonup$.
\begin{theorem}\mlabel{th:diff-action}
A coalgebra homomorphism $T:K\to H$ is a relative Rota-Baxter operator   if and only if $K$ endowed with the binary operation $*_T$ in \eqref{rrb-des1} is an algebra, denoted by $K_T=(K,*_T)$, and $H$ is a $K_T$-module
via the action $\star_T$ defined by
$$a\star_T x\coloneqq T(a)x,\quad\forall x\in H,a\in K.$$
\end{theorem}
\begin{proof}
If $T:K\to H$ is a relative Rota-Baxter operator, then by Corollary~\ref{prop:rrb-des}, $K_T=(K,*_T)$ is an algebra with unit 1. Also,
\begin{align*}
1\star_T x&=T(1)x=1x=x,\\
(a*_T b)\star_T x&=T(a*_T b)x=T(a)T(b)x=a\star_T(b\star_T x),
\end{align*}
for any $x\in H,\,a,b\in K$. That is, $H$ is a $K_T$-module.

Conversely, if $K_T=(K,*_T)$ is an algebra and $H$ is a $K_T$-module via the stated action $\star_T$, then particularly
\begin{align*}
T(a_1(T(a_2)\rightharpoonup b))
&=T(a*_T b)=T(a*_T b)1=(a*_T b)\star_T 1\\
&=a\star_T(b\star_T 1)=T(a)(T(b)1)=T(a)T(b).
\end{align*}
Namely,  \eqref{eq:rrb-Hopf} holds, and $T:K\to H$ is a relative Rota-Baxter operator.
\end{proof}


The following result is straightforward to obtain.
\begin{lemma}\mlabel{lemma:smash-modalg}
Let $H$ and $K$ be two cocommutative Hopf algebras such that $K$ is an $H$-module bialgebra via an action $\rightharpoonup$. Then $K$ is a $K\rtimes H$-module bialgebra defined by
\begin{equation*}
(a\#x).b\coloneqq \ad_a(x\rightharpoonup b),\quad \forall x\in H,\,a,b\in K.
\end{equation*}
\end{lemma}

 \emptycomment{
\begin{proof}
First note that $x\otimes y\otimes z\mapsto \ad_a(x\rightharpoonup b)$ is clearly a linear map from $K\otimes H^{\otimes 2}$ to $H$.
For any $x,y\in H,\,a,b,c\in K$, as both $H$ and $K$ are cocommutative, we have
\begin{align*}
(a\#x)(b\#y).c&=(a(x_1\rightharpoonup b)\#x_2y).c=\ad_{a(x_1\rightharpoonup b)}(x_2y\rightharpoonup c)\\
&=\ad_a [(x_1\rightharpoonup b_1)(x_2\rightharpoonup (y \rightharpoonup c))S_K(x_3\rightharpoonup b_2)]\\
&=\ad_a [(x_1\rightharpoonup b_1)(x_2\rightharpoonup (y \rightharpoonup c))(x_3\rightharpoonup S_K(b_2))]\\
&=\ad_a (x\rightharpoonup b_1(y \rightharpoonup c)S_K(b_2))\\
&=\ad_a (x\rightharpoonup (\ad_b(y \rightharpoonup c)))\\
&=(a\#x).((b\#y).c),\\
(1\#1).c&=\ad_1(1\rightharpoonup c)=c,\\
(a\#x).bc&=\ad_a(x\rightharpoonup bc)\\
&=\ad_{a_1}(x_1\rightharpoonup b)\,\ad_{a_2}(x_2\rightharpoonup c)\\
&=((a_1\#x_1).b)((a_2\#x_2).c),\\
(a\#x).1&=\ad_a(x\rightharpoonup 1)\\
&=\ad_a(\vep_H(x)1)=\vep_K(a)\vep_H(x)1,\\
\Delta_K((a\#x).b)&=\Delta_K(\ad_a(x\rightharpoonup b))\\
&=\ad_{a_1}(x_1\rightharpoonup b_1)
\otimes \ad_{a_2}(x_2\rightharpoonup b_2)\\
&=(a_1\#x_1).b_1\otimes(a_2\#x_2).b_2,\\
\vep_K((a\#x).b)&=\vep_K(\ad_a(x\rightharpoonup b))\\
&=\vep_K(a)\vep_H(x)\vep_K(b)=\vep(a\#x)\vep_K(b).
\end{align*}
So $K$ becomes a $K\rtimes H$-module bialgebra.
\end{proof}
}

 \begin{prop}
Let $T:K\to H$ be a relative Rota-Baxter operator. Then $K$ has a cocommutative $K_T$-module bialgebra structure via the following action,
$$\ad_{T,a}b\coloneqq \ad_{a_1}(T(a_2)\rightharpoonup b),\quad \forall a,b\in K.$$
\end{prop}
\begin{proof}
Let $T:K\to H$ be a relative Rota-Baxter operator. By Theorem~\mref{th:rrb-gr}, the graph $\Gr_T$ is a Hopf algebra inside $K\rtimes H$. Therefore $K$ becomes a $\Gr_T$-module bialgebra by Lemma~\mref{lemma:smash-modalg}.

Furthermore, pulled back by the Hopf algebra isomorphism $\Psi:K_T\to \Gr_T$ given in  \meqref{eq:isom-rbb-gr}, $K$ becomes a $K_T$-module bialgebra via the desired action $\ad_T$.
\end{proof}

\emptycomment{
\begin{remark}
 According to Theorem \ref{thm:restriction}, a relative   Rota-Baxter operator $T:K\to H$  induces a relative  Rota-Baxter operator $T:P(K)\to P(H)$. Consequently, we get the   descendent  Lie algebra  $P(K)_T$, where the Lie algebra structure $[\cdot,\cdot]_T$ on $P(K)$ is given by
 $$[a,b]_T=T(a)\rightharpoonup b -T(b)\rightharpoonup a+[a,b],\quad \forall a,b\in P(K).$$
    Restrict on $P(K)$ of the above new module algebra structure, we find that $P(K)_T$ acts on $P(K)$ by derivations via $\ad_T$ defined by
   $$
   \ad_T(a)(b)=[a,b]+T(a)\rightharpoonup b.
   $$
   Obviously, this fact holds for arbitrary (relative) Rota-Baxter operator  on Lie algebras. 
\end{remark}
}

\vspace{0.1cm}
 \noindent
{\bf Acknowledgements. } This research is supported by NSFC (Grant No. 11922110, 12071094, 12001228).

\bibliographystyle{amsplain}

\end{document}